\documentclass[letter,fullpage,11pt]{article}
\usepackage[english]{babel}
\usepackage[utf8]{inputenc}
\usepackage[hmargin={1in, 1in}, vmargin={1in, 1in}]{geometry}
\usepackage[dvips]{graphicx}
\usepackage{amsmath,amssymb,amsthm}
\usepackage{mathtools,xparse}
\usepackage{graphics}
\usepackage{fancyhdr}
\usepackage{float}
\usepackage{xcolor}
\usepackage{multicol}
\usepackage{accents}
\usepackage{makeidx} 
\usepackage{enumitem}
\usepackage{subcaption}
\usepackage{amsthm}
\usepackage{setspace,caption}
\usepackage{lipsum}
\usepackage{environ}
\usepackage{epsfig}
\usepackage[T1]{fontenc}
\usepackage[applemac]{inputenx}
\usepackage{multirow}
\usepackage{bbm}
\usepackage{tikz}
\usetikzlibrary{positioning}
\usepackage{bm}

\usepackage[authoryear]{natbib}

\theoremstyle{dotless}

\newtheorem{theorem}{Theorem}

\newtheorem{assumption}{Assumption}

\newtheorem{lemma}{Lemma}

\newlength{\dhatheight}
\newlength{\dtildeheight}

\NewEnviron{myequation}{%
    \begin{equation*}
    \scalebox{0.8}{$\BODY$}
    \end{equation*}
    }
   
\newcommand\numberthis{\addtocounter{equation}{1}\tag{\theequation}}
\newcommand*\diff{\mathop{}\!\mathrm{d}}

\newcommand{\Rn}[1]{%
  \textup{\uppercase\expandafter{\romannumeral#1}}%
}
\newcommand\independent{\protect\mathpalette{\protect\independenT}{\perp}}
\def\independenT#1#2{\mathrel{\rlap{$#1#2$}\mkern2mu{#1#2}}}
 
\DeclarePairedDelimiterX{\norm}[1]{\lVert}{\rVert}{#1}
\DeclarePairedDelimiterX{\abs}[1]{\lvert}{\rvert}{#1}

\newcommand{\stwo}[1]{{#1}}   

\allowdisplaybreaks 

\def\hat{\widehat}
\def\SS{{\cal S}}
\def\dbar{\ \bar{}\! \! d}
\def\mbar{\bar{m}}

\makeatletter
\newcommand{\subalign}[1]{%
  \vcenter{%
    \Let@ \restore@math@cr \default@tag
    \baselineskip\fontdimen10 \scriptfont\tw@
    \advance\baselineskip\fontdimen12 \scriptfont\tw@
    \lineskip\thr@@\fontdimen8 \scriptfont\thr@@
    \lineskiplimit\lineskip
    \ialign{\hfil$\m@th\scriptstyle##$&$\m@th\scriptstyle{}##$\crcr
      #1\crcr
    }%
  }
}
\makeatother

\renewenvironment{abstract}
 {\small
  \begin{center}
  \bfseries \abstractname\vspace{-.5em}\vspace{0pt}
  \end{center}
  \list{}{
    \setlength{\leftmargin}{1cm}%
    \setlength{\rightmargin}{\leftmargin}%
  }%
  \item\relax}
 {\endlist}

\title{}
\author{}

\begin{document}

\setlength{\abovedisplayskip}{2.5pt}
\setlength{\belowdisplayskip}{2.5pt}


\thispagestyle{empty}

\begin{center}
{\Large \bf Nonparametric Estimation of Conditional Expectation with Auxiliary Information and Dimension Reduction}	
\end{center}
\begin{center}
Bingying Xie and Jun Shao\\
Department of Statistics, University of Wisconsin-Madison
\end{center}
\begin{abstract}
Nonparametric estimation of the conditional expectation $E(Y | U)$ of an outcome $Y$ given a covariate vector $U$  is of primary importance in many statistical applications such as prediction and personalized medicine. In some problems, there is an additional  auxiliary variable $Z$ in the training dataset used to construct estimators, but $Z$ is not available for future prediction or selecting patient treatment  in personalized medicine. For example, in the training  dataset  longitudinal outcomes are observed, but only the last outcome $Y$ is concerned  in the future prediction or analysis. The longitudinal outcomes other than the last point is then the variable $Z$ that is observed and related with both $Y$ and $U$. Previous work on 
how to make use of $Z$ in the estimation of $E(Y|U)$ mainly focused on using $Z$ in the construction of a linear function of $U$ to reduce covariate dimension for better estimation. Using $E(Y|U) = E\{E(Y|U, Z)| U\}$, we propose 
a two-step estimation of inner and outer expectations, respectively, with sufficient dimension reduction for kernel estimation in both steps. The information from $Z$ is utilized not only in dimension reduction, but also directly in the estimation. 
Because of the existence of different ways for dimension reduction, we construct two estimators that may improve the estimator without using $Z$. The improvements are shown in  the convergence rate of estimators as the sample size increases to infinity as well as in the finite sample simulation performance.
A real data analysis about the selection of mammography intervention  is presented for illustration. \\
\end{abstract}


\newpage
\section{Introduction}

In many statistical applications, a key step is  to estimate the  conditional expectation of $Y$ given $U= u_0$, denoted as $\psi (u_0) = E(Y | U=u_0)$ in what follows, based on  a training  sample dataset,  
where $Y$ is a response of interest, $U$ is a vector of covariates, and $u_0$ is a given specific value of $U$. Apparently, the prediction of a future $Y$ at $U=u_0$  is an example.  Another example is in the area of personalized medicine in which we would like to maximize the condition expectation $E(Y|U=u_0, a)$ over several treatment options $a=1,...,k$ \citep{qian_murphy_2011}, where 
$u_0$ is the vector of a future patient's prognostic 
factors and  demographic variables, and $Y$ is his or her future outcome. Larger (or smaller) $Y$ means better outcome. Because parametric modeling of $\psi (u_0)$ is difficult in many applications such as the personalized medicine problems, nonparametric kernel estimation of $\psi (u_0)$  
\citep{nadaraya_1964, watson_1964}
has been widely considered and used. 
As shown in Theorem 2.2.2 of \citet{bierens},  the optimal convergence rate of a kernel estimator is $n^{-m/(2m+p)}$, where $m$ is the order of kernel and $p$ is the dimension of $U$.   When $p$ is not small, it is crucial to search for a 
matrix $B$ with the smallest possible column dimension $d_0 < p $ such that 
 $E(Y|U)=E(Y|B^{T}U)$,  where $B^T$ is the transpose of $B$, and hence the optimal convergence rate is improved to  $n^{-m/(2m+d_0)}$.
 This is usually achieved by using 
the training data to estimate a $B$ with smallest  column dimension such that
$Y\independent U \mid B^T U$, i.e., $Y$ and $U$ are independent conditional on $B^TU$,  which is referred to as sufficient dimension reduction   (SDR)
\citep{li_1991, cook_weisberg_1991, xia_tong_li_zhu_2002, li_wang_2007, ma_zhu_2012}. The linear space generated by the columns of $B$  is called  the central subspace for $Y$ given $U$ and is denoted as ${\cal S}(B)= {\cal S}_{Y|U}$. 

 Besides $U$, in some problems there exists a vector $Z$ of 
auxiliary variables in the training sample dataset, but $Z$ is not available in the future. For example, in many clinical or observational studies, covariate $U$ and longitudinal responses $Y_1,...,Y_T$ are observed in the training dataset, where $Y_t$ is the response at time $t$, but in the future, we may only observe $U=u_0$ 
 to predict $Y=Y_T$ at time $T$ without  the additional 
   $Z = (Y_1,...,Y_{T-1})$.  In many situations, $Z$ may be more related with $Y$ than $U$.    
   This raises an issue of how to make use of the available data in $Z$ in the training dataset to improve the estimation of $\psi (u_0)$. 
   
 
 
 Efforts have been made in utilizing $Z$ data in SDR.
For a discrete $Z$ taking values $z_1,...,z_L$,  \citet{chiaromonte_cook_li_2002}
   proposed to use  ${\cal S}_{Y|U}^Z = {\cal S}(B_{z_1})
   \oplus \cdots \oplus {\cal S}(B_{z_L})$, where 
    $B_{z_l}$  has the smallest column dimension such that $
Y\independent U \mid  B_{z_l}^T U, Z= z_l$ and 
 ${\cal S}_1 \oplus {\cal S}_2
= \{s_1+s_2 :  s_{j}\in {\cal S}_j, j=1,2 \}$.
However, they cannot guarantee that  ${\cal S}_{Y|U}^Z $
coincides with the central subspace ${\cal S}_{Y|U} = {\cal S}(B)$, 
  $Y\independent U \mid B^T U$.
  \citet{hung_liu_lu_2015} proposed a two-stage method of searching $B$ in a $Z$-envelope  $= {\cal S}_{Y|U}^Z  \oplus  {\cal S}_{Z|U} \supseteq {\cal S}(B)$, where  ${\cal S}_{Z|U}$ is the central subspace for $Z$ given $U$. 
Although their method utilizes $Z$ data to produce a better $B$ estimator,
the resulting  estimator of $\psi(u_0)$ has the same convergence rate as the estimator  based on an estimator of $B$ without using $Z$ data, because a better estimator of $B$ does not improve the convergence rate of the estimator of $\psi(u_0)$. 

Instead, in this paper we propose an idea of using $Z$ data in the estimation of $\psi (u_0)$ directly, based on the following well known identity:
\begin{equation}
\psi (u_0) = E(Y|U=u_0) = E\{ E(Y|Z, U=u_0) | U=u_0 \} \label{identity}
\end{equation}
We utilize the $Z$ information in the estimation of  inner expectation $E(Y|Z,U)$
treating $Z$ as a part of covariate  as well as outer expectation $E\{ \, \cdot \, | U=u_0\} $ using the conditional distribution of $Z$ given $U$. SDR is applied in the kernel estimation of both expectations and is necessary because incorporating $Z$ data increases the dimensions of kernels in kernel  estimation. 

We consider two ways of reducing dimensions, which lead to two different  estimators of $\psi (u_0)$. The first method  performs  SDR to find a matrix $C_{zu}= \left(^{C_z}_{C_u}\right)$
so that the inner expectation 
is  $E(Y|Z,U) = E(Y | C_{zu}^T \big(_U^Z\big))=
E(Y|C_z^T Z +C_u^T U)$,  and then another SDR to find a matrix $C$ with $C_z^TZ \independent U \mid C^TU$ which implies  $E(Y|U) =E\{ E(Y|C_z^T Z+ C_u^T U)|C^TU\} $. We show that the convergence rate of kernel estimator of $\psi (u_0)$ using this method is 
$n^{-m/(2m+d_1)}$ depending on $d_1$, the column dimension of $C$,
 not the column dimension of $C_{zu}$. Although the column dimension of $C_{zu}$ does not affect the convergence rate, reducing $(Z,U)$ to $C_z^T Z +C_u^T U$ is still important for kernel estimation of the inner expectation. 
 
 However, it is not always true that $d_1\leq d_0$, the dimension of the central subspace $\SS_{Y|U}= \SS(B)$ without $Z$. Thus,  the estimator of $\psi(u_0)$ based on the first method does not always improve the estimator without using $Z$. To ensure obtaining an estimator with a convergence rate no slower than that of the estimator without using $Z$ data, we propose the second method using SDR
  to find the following matrices: (i)  a matrix $B$ satisfying $E(Y|U)= E(Y|B^TU)$; (ii) a matrix $D_{zu} = \left(^{D_z}_{D_u}\right)$ 
 satisfying  
$E(Y|Z,B^TU)= E(Y| D_{zu}^T \big(_{B^TU}^{~~ Z}\big)) = E(Y|D_z^TZ +D_u^TB^TU)$; and (iii) a matrix $D$ satisfying 
$D_z^TZ \independent B^TU \mid D^TB^TU$.   
Then, 
\[ E(Y|U)= E(Y|B^TU) = E\{ E(Y| D_z^TZ+ D_u^TB^TU)|D^TB^TU\} . \]
We show that  the convergence rate of kernel estimator of 
$\psi(u_0)$ using this method is \linebreak 
$n^{-m/(2m+d_2)}$ depending on $d_2$, the column dimension of $D$. 
By applying SDR, it is 
guaranteed that $d_2 \leq d_0$. In fact, $d_2 <d_0$ in many situations. For the first method, in some situations $d_1$ can be even smaller than $d_2 \leq d_0$, although it does not guarantee $d_1 \leq d_0$. See Examples 1-3 in Section 2. 

Why can we improve the convergence rate in estimating $\psi(u_0)$? Without $Z$ data, the best we can do is to use the central subspace $\SS_{Y|U} = \SS(B)$ whose  dimension determines the convergence rate. \citet{hung_liu_lu_2015} utilized $Z$ data to improve the estimation of $B$, but they could not improve the convergence rate. 
However, our approach is to use formula (\ref{identity}) and estimate $\psi(u_0)$ in two steps, the estimation of inner and outer expectations, with SDR in both steps. Because the convergence rate depends on the convergence rate of outer expectation estimation involving $Z$ given $U$, we may be able to make use of a space that is smaller than $\SS(B)$, e.g., the space generated by columns of $BD$ in our second method, which cannot be achieved without the inner expectation estimation involving $Z$ data. 
 
Details of our proposed estimation procedures are given in Section 2 with three examples for illustration. 
In Section 3, we establish the asymptotic normality of proposed estimators under some regularity conditions, and  obtain the optimal convergence rates and the asymptotic mean squared errors. 
Simulation studies under  various circumstances are considered in Section 4. A real data analysis about the selection of mammography intervention methods is carried out in  Section 5 to illustrate our procedure.  All the technical proofs are given in the Appendix. 
	
\section{Methodology}

Throughout  we use  $K_h$ as a generic notation for a  kernel with an appropriate dimension and bandwidth $h$, i.e., $K_h$ appeared in different places may be different. Assumptions on the kernels are introduced in Section 3. 
Let $\{ Y_i,U_i,Z_i, i=1,...,n\}$ be an independent and identically distributed training sample of size $n$  from $(Y,U,Z)$.  
Without using $Z$ data and dimension reduction, 
a  kernel regression estimator of $\psi (u_0)$ defined in (\ref{identity}) is 
\begin{align*} 
\hat{\psi}_{p}(u_0)=\sum_{i=1}^{n}Y_{i}K_{h}(U_{i}-u_{0})
\bigg/ \sum_{i=1}^{n}K_{h}(U_{i}-u_{0})
\end{align*}
where the subscript $p$ indicates that $\hat\psi_p$ uses a kernel with dimension $p$, the dimension of $U$. 

Suppose that $Y\independent U \mid B^TU$, where $B$ has the smallest column dimension $d_0 \leq p$.  The estimator 
$\hat{\psi}_p(u_0)$ can be improved by 
\begin{align} \label{esti2}
\hat{\psi}_{d_0}(u_0)=\sum_{i=1}^{n}Y_{i}K_{h}(\hat{B}^T U_i -\hat{B}^T u_{0})\bigg/ \sum_{i=1}^{n}K_{h}(\hat{B}^T U_i -\hat{B}^T u_{0})
\end{align}
where  $\hat{B}$ is an estimator of $B$ by SDR. 


To make use of the auxiliary information provided by $Z$, we use identity 
(\ref{identity}) and first estimate the inner expectation 
$ E(Y|Z, U)$. 
Following the discussion in Section 1, 
we construct SDR estimators $\hat{C}_z$ and $\hat{C}_u$ of
  $C_z$ and $C_u$, respectively, with  $E(Y|Z,U)= E(Y|C_z^T Z + C_u^T U)$. Then $E(Y|Z,U)$ can be estimated by \vspace{-1mm}
 \begin{align}\label{esti31}
\begin{split}
\hat{\varphi}_1(Z,U) 
=\sum_{j=1}^{n}Y_{j}K_{h}
( \hat{C}_z^TZ_j + \hat{C}_u^T U_j  -\hat{C}_z^TZ - \hat{C}_u^T U ) \\
\! \bigg/  \!
\sum_{j=1}^{n}K_{h}
( \hat{C}_z^TZ_j + \hat{C}_u^T U_j  -\hat{C}_z^TZ - \hat{C}_u^T U )
\end{split}
\end{align}
For the second step of estimating the outer expectation in (\ref{identity}), we construct an SDR estimator $\hat{C}$ of  $C$ satisfying  $C_z^TZ \independent U \mid C^TU$.  Then our first proposed  estimator of $\psi(u_0)$ is 
\begin{align} \label{esti3}
\hat{\psi}_{d_1}(u_0)&=\sum_{i=1}^{n} \hat{\varphi}_1( Z_i , u_0)
K_{h}(\hat{C}^TU_{i}-\hat{C}^Tu_0)\bigg/
\sum_{i=1}^{n}K_{h}(\hat{C}^T U_{i}-\hat{C}^Tu_0)
\numberthis 
\end{align}
where $d_1$ is the column dimension of $C$, which   
is not necessarily smaller than $d_0$, the column dimension of $B$.  Thus, $\hat\psi_{d_1}(u_0)$ is not always better than $\hat\psi_{d_0}(u_0)$ 
in terms of convergence rate established in Section 3. 





To derive an estimator having convergence rate no slower than that of $\hat\psi_{d_0}(u_0)$, we build the improvement on $(Y,Z,B^TU)$, instead of 
$(Y,Z,U)$ in the derivation of $\hat\psi_{d_1}(u_0)$, since $Y \independent U \mid B^TU$. We still use (\ref{identity}) to do estimation in two steps. 
The first step is the same as the first step of constructing $\hat\psi_{d_1}(u_0)$ in (\ref{esti3}) except that $U$ is replaced by $B^TU$. That is, we construct 
 SDR estimators $\hat{D}_z$ and $\hat{D}_u$ of
$D_z$ and $D_u$, respectively, with  $E(Y|Z,B^TU)= E(Y|D_z^T Z + D_u^T B^TU)$, and estimate 
$E(Y|Z,B^TU)$  by 
\begin{align*}
 \hat{\varphi}_2(Z,\hat B^TU)  
 =  \sum_{j=1}^{n}Y_{j}K_{h}
( \hat{D}_z^TZ_j \! + \! \hat{D}_u^T \hat B^T U_j  \!-\! \hat{D}_z^TZ \!-\! \hat{D}_u^T \hat B^T U ) \\ 
\! \bigg/ \!
\sum_{j=1}^{n}K_{h}
( \hat{D}_z^TZ_j \!+\! \hat{D}_u^T \hat B^T U_j \! -\!\hat{D}_z^TZ \!-\! \hat{D}_u^T \hat B^T U )
\end{align*}
where $\hat{B}$ is defined in (\ref{esti2}).
For the second step, we construct an SDR estimator $\hat D$ of $D$ satisfying 
$D_z^TZ \independent B^TU \mid {D}^TB^TU$. 
Then, our second proposed estimator of $\psi(u_0)$ is 
\begin{align} \label{esti4}
\hat{\psi}_{d_2}(u_0)&=\sum_{i=1}^{n} \hat{\varphi}_2( Z_i , \hat B^T u_0)
K_{h}(\hat{D}^T\hat{B}^T U_{i}-\hat{D}^T\hat{B}^Tu_0)\bigg/
\sum_{i=1}^{n}K_{h}(\hat{D}^T \hat{B}^TU_{i}-\hat{D}^T\hat{B}^Tu_0)
\end{align}
where $d_2$ is the column dimension of $D$. Note that 
$d_2 \leq d_0$, the column dimension of $B$. 

The next lemma shows the relationship among the spaces generated by $B$, $C$, and $D$. \vspace{-1mm}
\begin{lemma} 
	\begin{itemize}
		\item[(i)]
If	$Y \independent U \mid B^TU$, $Y \independent (Z,U) \mid C_z^TZ+C_u^TU$, and 
$C_z^TZ \independent U \mid C^TU$, where $B$, $C_z$, $C_u$, and $C$ all have the smallest possible column dimensions, then
		$\SS (B) \subseteq \SS (C_u)\oplus \SS (C)
	$.
		  If, in addition,  $C_z^T Z \independent U | B^T U$, then $\SS (C) \subseteq \SS (B)$.\vspace{-2mm}
		\item[(ii)]
		If $Y \independent U \mid B^TU$, $Y \independent  (Z, B^TU ) \mid D_z^TZ+D_u^TB^T U$, and 
		$D_z^TZ \independent B^TU \mid D^TB^TU$, where $B$, $D_z$, $D_u$, and $D$ all have the smallest possible column dimensions, then
	$\SS (B) 	=\SS (BD_{u}) \oplus \SS ( BD)$.  
\end{itemize}
\end{lemma}

The result in Lemma 1(i) says that $C$ may 
contain column vectors  that do not belong to $\SS (B) $, unless 
 $C_z^T Z \independent U | B^T U$. In general, $\SS (C)$ and $\SS (B)$ may not have any relationship so that $\hat\psi_{d_1}$ may be more or less efficient than $\hat\psi_{d_0}$. On the other hand, Lemma 1(ii) indicates that  
$D_u$ and $D$ used in two steps  together support the central subspace $\SS (B)$, and $\hat\psi_{d_2}$ is much more efficient than $\hat\psi_{d_0}$ if 
$\SS (BD)$ is truly contained in $\SS (B)$, otherwise 
$\hat\psi_{d_2}$ and $\hat\psi_{d_0}$ have the same convergence rate.
 Thus, $\hat\psi_{d_2}$ is guaranteed to be more efficient than or as good as $\hat\psi_{d_0}$ in terms of convergence rate. Regarding $\hat\psi_{d_1}$ and $\hat\psi_{d_2}$, there is no definite conclusion about their relative efficiency, since $\SS (C)$ and $\SS (D)$ do not have relationship. 

Three examples are provided next for illustration
on why and when the proposed estimator $\hat\psi_{d_1}$ or $\hat\psi_{d_2}$ is better than other estimators.
\vspace{2mm}

\noindent
{\bf Example 1}. 
Suppose that $U$ consists of 4  components
$ u_1,u_2,u_3$, and $u_4$, and that random variables $ u_1,u_2,u_3, u_4$, $\eta_1, \eta_2, \eta_3$, and $\epsilon$ are mutually independent, and $E(\epsilon )=0$. Assume also that     
$Z$ has 3 components,  
$			z_1= \abs{u_1-u_2}+\eta_1$, 
$			z_2=u_2+\eta_2$, and 
$			z_3=u_1+\eta_3 $, and that 
\[ 			Y=(z_1+7u_4)(u_3^2+1)^{-1}+\epsilon 
			  = (\abs{u_1-u_2}+\eta_1+7u_4)(u_3^2+1)^{-1}+\epsilon. \] 
		A straightforward calculation gives 
	\[ B = \left( \begin{array}{rrr}
	1 & 0 & 0 \\
	-1 & 0 & 0 \\
	0 & 1 & 0 \\
	0 & 0 & 1 \end{array} \right) \ \
	C_z =D_z = \left(  \begin{array}{rr} 
	1 & 0 \\
	0 & 0 \\
	0 & 0 \end{array} \right) \ \
	C_u = BD_u = \left(  \begin{array}{rr} 
	0 & 0 \\
	0 & 0 \\
	0 & 1 \\
	7 & 0  \end{array} \right) \ \
	C= BD = \left(  \begin{array}{r}	
	1 \\
	-1 \\
	0 \\
	0 \end{array} \right) \]

In this example, $p=4$, $d_0=3$, and $d_1 = d_2 = 1$. 
When $Z$ is not considered, $\hat\psi_{d_0}$ improves $\hat\psi_p$ since $d_0<p$. 
$\hat\psi_{d_1}$ and $\hat\psi_{d_2}$ are identical in this example, and they 
are  more efficient than $\hat\psi_{d_0}$ since $d_1=d_2=1 < d_0=3$. 

Here is an explanation on why our method improves $\hat\psi_{d_0}$ in this particular example. 
In the first step of estimating the inner expectation in (\ref{identity}),  
$Y$ is found to be related with two variables $z_1+7u_4$ and $u_3$;
in the second step of estimating the outer expectation in (\ref{identity}), 
$C_z^TZ = D_z^TZ=z_1$ is found to be related with one variable $u_1-u_2$. 
Thus, our approach  ``splits'' the original task of estimating $E(Y|U)$ with three variables into 
two tasks, estimating the inner expectation with two variables and estimating the outer expectation with one variable. It is shown in Section 3 that the convergence rate of $\hat\psi_{d_1}$ or 
$\hat\psi_{d_2}$ depends on the kernel estimation of the outer expectation and, consequently, this split produces an estimator with a faster convergence rate.  

It is also interesting to notice that $\SS (B)=\SS (C_u)\oplus \SS (C)$ in this example, i.e.,  the existence of $Z$  splits $\SS (B)$ into two orthogonal spaces and $Z$ does not bring in any unwanted information outside of $\SS (B)$. 
\vspace{2mm}

\noindent{\bf Example 2}.
Consider the same setting as in Example 1 except that
$	z_1=-5(u_2-\eta_1)+0.1u_1u_3$, 
$	z_2=0.5|u_1|+\eta_2$, 
$	z_3=-3(|u_2|-\eta_3)$, and 
$$	Y =  z_1-0.1u_1u_3-3z_3+|u_4| + 0.5 \epsilon=-5(u_2-\eta_1)+9(|u_2|-\eta_3)+|u_4|+0.5\epsilon.$$ 
Then
	\[ B = \left( \begin{array}{rrr}
0 & 0  \\
1 & 0  \\
0 & 0  \\
0 & 1  \end{array} \right) \quad 
C_z = \left(  \begin{array}{rrrr} 
1 & 0 & 0 & 0 \\
0 & 0 & 0 & 0 \\
-3 & 0 & 0 & 0 \end{array} \right) \quad
C_u =  \left(  \begin{array}{rrrr} 
0 &  1 & 0 & 0 \\
0 & 0 & 0 & 0\\
0 & 0 & 1 & 0 \\
0 & 0 & 0 & 1  \end{array} \right) \quad
C= \left(  \begin{array}{rrr}	
0 & 1 & 0 \\ 
1 & 0 & 0 \\
0 & 0 & 1 \\
0 & 0 & 0 \end{array} \right) \]
Without using $Z$, $\hat\psi_{d_0}$ involves $u_2$ and $u_4$ 
and $d_0=2$. Note that $u_1u_3$ is useful for $Z$ but not for $Y$. 
As the variable $u_1u_3$ outside the space of $\SS (B)$ is redundantly brought into the estimation,  $\hat\psi_{d_1}$ uses four variables $z_1-3z_3$, $u_1$, $u_3$, and $u_4$ in the first step and three variables 
$u_1$, $u_2$, and $u_3$ in the second step. Thus,  
 $d_1 = 3$ and $\hat\psi_{d_1}$
 is even less efficient than $\hat\psi_{d_0}$.  

On the other hand, the construction of $\hat\psi_{d_2}$ starts with  
 $(Y,Z,B^TU)$ so that  $u_1u_3$ is never in the picture, since $u_1u_3$ is independent of $(u_2,u_4)$. Note that 
 \[ D_z = \left(  \begin{array}{rrrr} 
 1  & 0  \\
 0 & 0 \\
 -3 & 0 \end{array} \right) \quad\quad
 D_u =  \left(  \begin{array}{rrrr} 
 0 & 0 \\
  0 & 1  \end{array} \right) \quad\quad
  D= 
  \left(  \begin{array}{rrr}	
  1 \\
  0 \end{array} \right) \quad\quad 
BD =  \left(  \begin{array}{rrr}	
 0 \\ 
 1 \\
 0 \\
 0 \end{array} \right) \]
because  $D_z^TZ=z_1-3z_3=-5u_2+0.1u_1u_3+9|u_2|+5\eta_1-9\eta_3$ is only related to $u_2$ in $B^TU= (u_2,u_4)^T$, not $u_4$, in the second step of the estimation, and $u_1u_3$ is independent of $(u_2,u_4)$. Hence,  ${d}_2=1<d_0=2<d_1=3$, and estimator $\hat\psi_{d_2}$ outperforms all other estimators in this example.
Again, $\hat\psi_{d_2}$ splits the task of estimating $E(Y|U)$ into two steps, with two variables $z_1-3z_3$ and $u_4$ in the first step and one variable $u_2$ in the second step. 
\vspace{2mm}

\noindent
{\bf Example 3}. 
This is an example in which $\hat\psi_{d_1}$ beats  $\hat\psi_{d_2}$ and $\hat\psi_{d_0}$. Consider the same setting as in Example 1 except that 
$	z_1= 2(-u_1+u_4)+\eta_1$, 
$	z_2=(-u_1+u_4)+\eta_2$, 
$	z_3=u_4+\eta_3$, and 
$$	Y  = z_1-z_2+u_1+u_2+(u_1-u_3)^2 + \epsilon =u_2+u_4+(u_1-u_3)^2+\eta_1-\eta_2+\epsilon$$
In the first step of $\hat\psi_{d_1}$, $Y$ is related with two variables 
$z_1-z_2+u_1+u_2$ and $u_1-u_3$, and in the second step, 
$C_z^TZ=z_1-z_2$ is a function of one variable $-u_1+u_4$. Hence 
	\[ B = \left( \begin{array}{rrr}
0 & 1  \\
1 & 0  \\
0 & -1  \\
1 & 0  \end{array} \right) \quad \quad
C_z = \left(  \begin{array}{rrrr} 
1 & 0 \\
-1  & 0 \\
0  & 0 \end{array} \right) \quad\quad
C_u =  \left(  \begin{array}{rrrr} 
1 & 1   \\
1 & 0 \\
0 & -1  \\
0 & 0   \end{array} \right) \quad\quad
C= \left(  \begin{array}{rrr}	
-1 \\ 
0 \\
0 \\
1 \end{array} \right) \]
and    $d_1=1<d_0=2$. 

For $\hat\psi_{d_2}$, we search directions in $(Z,B^TU)$ in the first step.
Note that $Y$ is related to  $z_1-z_2+u_1+u_2$ and $u_1-u_3$.
Although $u_1-u_3$ is exactly the second component of $B^TU$, one cannot express  $z_1-z_2+u_1+u_2$ as a linear function of $Z$ and $B^TU$. Hence, 
	\[ 
D_z = \left(  \begin{array}{rrrr} 
1 & 0 & 0\\
-1  & 0 & 0 \\
0  & 0 & 0 \end{array} \right) \quad\quad
D_u =  \left(  \begin{array}{rrrr} 
0 & 1 & 0  \\
0 & 0 & 1   \end{array} \right) \]
In the second step, $D_z^TZ=z_1-z_2=-u_1+u_4+\eta_1-\eta_2$. 
Still, one cannot find a linear function of $B^TU$ to 
 represent the vector $(-1 \  0 \ 0 \ 1)^T$ related to $-u_1+u_4$. Thus, 
we have to use both column vectors of $B$ despite $-u_1+u_4$ is one dimensional.  
Then $D$ is the identity of order 2 and $BD=B$, 
i.e., in the second step we cannot reduce the dimension of $B^TU$ any more. 
Since $d_2=2>d_1=1$, $\hat\psi_{d_1}$  turns out to be better than $\hat\psi_{d_2}$, and $\hat\psi_{d_2}$ has the same convergence rate as $\hat\psi_{d_0}$. 
 
 This example shows that restricting to $B^TU$ may prevent us to find the best direction in $B^TU$ for the outer expectation estimation, although it guarantees that at least we use $B^TU$ so that the resulting estimator is at least as good as $\hat\psi_{d_0}$. 
 
 \section{Asymptotic Properties}

%
This section is dedicated to the asymptotic properties of the estimators of $\psi (u_0)$ formulated in Section 2. 
It is clear that the asymptotic properties of $\hat\psi_{d_0}(u_0)$, $\hat\psi_{d_1}(u_0)$, and $\hat\psi_{d_2}(u_0)$ depend on the asymptotic convergence rates of the SDR estimators $\hat{B}$, $\hat{C}_z$, 
$\hat{C}_u$, $\hat{C}$,  $\hat{D}_z$, 
$\hat{D}_u$, and $\hat{D}$  in (\ref{esti2})-(\ref{esti4}). Under reasonable conditions, it is  proved in \citet{ma_zhu_2012} that 
 SDR estimators converge at the rate $n^{-1/2}$, which is assumed throughout this paper. 

In the beginning of Section 2, we introduced  a generic notation $K_h$ for a kernel with  bandwidth $h$. 
In what follows  $K_h$ is chosen to be a product kernel of dimension $s$ and  order $m \geq 2$  in the sense that  
$K_{h}(x)=h^{-s} \prod_{j=1}^{s}\kappa(x_j/h)$, where  $x_j$ is the $j$th component of the $s$-dimensional $x$  and  $\kappa( \cdot )$ is a bounded and  Lipschitz continuous univariate kernel having a compact support  and satisfying $\int \kappa (t) dt =1$,  
 $\int t^m \kappa (t) dt $ is finite and nonzero, and 
 $\int t^l \kappa (t) dt = 0$ for all $0< l <m$. 
 
For $\hat{\psi}_{d_1}$, let    $V=C_z^TZ+C_u^TU$,  
$\gamma(v)
$ be the two dimensional vector whose components are  
$\gamma_1(v)= f_V(v)$ and 
$\gamma_2(v)= E(Y|V=v)f_V(v)$, and 
$\hat{\gamma}(v)$  be the two dimensional vector whose components are 
$\hat{\gamma}_1(v)= n^{-1}\sum_{j=1}^n K_h(C_z^TZ_j+C_u^TU_j - v)$ and 
$\hat{\gamma}_2(v) =  n^{-1} \sum_{j=1}^n Y_jK_h(C_z^TZ_j+C_u^TU_j -v)$. 
For $\hat{\psi}_{d_2}$, $V$ takes the form $V=D_z^T Z+D_u^T B^T U$ and 
$\hat\gamma_i(v)$ is defined with $C_z^TZ_j+C_u^TU_j$ replaced by 
$D_z^TZ_j+D_u^TB^TU_j$. 



Throughout, we use $f_X( \cdot )$ to denote the probability density of a random vector $X$. 

\begin{assumption}\label{bound} 
	The density $f_{V}$ is bounded below from zero, i.e., there is a constant $ c>0$ such that  $\inf_{v} f_{V}(v) \geq c$.
\end{assumption}

We state the following assumptions  for $\hat\psi_{d_1}$. For $\hat\psi_{d_2}$, the assumptions should be modified as in the statement of Theorem 1. 
To simplify expressions in assumptions and theorem, for both $\hat{\psi}_{d_1}$ and $\hat{\psi}_{d_2}$, we use the same notation
 $\dbar$,  $\hbar$ and $\mbar$ to denote the dimension, bandwidth, and order of the kernel used in the 
  inner expectation estimation, and $\stwo{d}$, $\stwo{h}$ and $\stwo{m}$ to  denote the  dimension, bandwidth, and order of the kernel  used in the outer expectation estimation.

\begin{assumption}\label{UnifConv}

The function  
$\gamma (v)$ has bounded  $\mbar$th derivative. The kernel bandwidth $\hbar$ of the first step is  of the order $n^{-\bar s}$ and there exists a $q >1$  such that 
$\bar s<(1-q^{-1})\dbar^{-1}$ and the function $E[ (1+Y^2)^{q}|V=v]$ is bounded.
\end{assumption}

Assumptions \ref{bound} and  \ref{UnifConv}  are similar to those in 
	 \citet{newey_1994} and  \citet{hansen_2008}, which 
ensure the convergence of $\hat{\gamma}(v)$ to $\gamma(v)$ uniformly in $v$. Specifically, $\norm{\hat{\gamma}(v)-E[\hat{\gamma}(v)]}_{\infty}=O_{p}(a_{n})$ and
	$\norm{E[\hat{\gamma}(v)]-\gamma (v)}_{\infty}=O_{p}(\hbar^{\mbar})$, where 
$a_{n}=(\log n / n \hbar^{\dbar})^{1/2}$ and $\| \cdot \|_\infty$ is the sup-norm; hence, 
$\norm{\hat{\gamma}(v)-\gamma(v)}_{\infty} = O_{p}(a_{n}+\hbar^{\mbar})$. 

\begin{assumption}\label{main}
Let  $\tilde{V}= \tilde{C}_z^{T}Z+\tilde{C}_u^{T}U$, and $\Omega=\{(\hat{\gamma},\tilde{C}_{zu},\tilde{C}): \norm{\hat{\gamma}(\tilde{v})-\gamma(v)}_{\infty}\leq \epsilon, \norm{\tilde{C}_{zu}-C_{zu}}\leq cn^{-1/2},\norm{\tilde{C}-C}\leq cn^{-1/2}\}$ for some positive constants $c$ and  $\epsilon>0$. 
\begin{enumerate}[label=(\roman*)]

\item\label{MaZhu1}
Uniformly on $\Omega$, the $\mbar$th derivative of $f_{\tilde{V}}(\tilde{v})$  and
$E(Y|\tilde{V}=\tilde{v})f_{\tilde{V}}(\tilde{v})$
are Lipschitz-continuous functions of $\tilde{v}$, 
 the $\stwo{m}$th derivatives of  
 $E[\hat{\gamma}_2(\tilde{V})/\hat{\gamma}_1(\tilde{V})|\tilde{C}^{T}U=t]f_{\tilde{C}^T U}(t)$ and $f_{\tilde{C}^T U}(t)$
 are Lipschitz-continuous functions of $t$,  
and 
$E[\hat{\gamma}_2^2(\tilde{V})/\hat{\gamma}_1^2(\tilde{V})|\tilde{C}^TU=
t ]$ and $E(Y^2|\tilde{V}=\tilde{v} )$ are  Lipschitz-continuous as functions of $t$ and $\tilde{v}$  respectively. 

\item\label{MaZhu2}
Uniformly on $\Omega$, $E(Y^2|\tilde{V}=$ $\tilde{C}_{z}^{T}z+\tilde{C}_{u}^{T}u)$ 
and $E[\hat{\gamma}_2^2(\tilde{V})/\hat{\gamma}_1^2(\tilde{V}) | \tilde{C}^{T}U=\tilde{C}^{T}u]$
 are bounded.
 
 \item\label{SDRneg}
Uniformly on $\Omega$, 
$E[\hat{\gamma}_2( \tilde{V})/\hat{\gamma}_1( \tilde{V}) | \tilde{C}^{T}U=\tilde{C}^{T}u_0]f_{\tilde{C}^{T}U}(\tilde{C}^{T}u_0)$ and 
$f_{\tilde{C}^{T}U}(\tilde{C}^{T}u_0)$ are  Lipschitz continuous  functions of $\tilde{C}$, and  
$E[ \norm{\partial E(Y|\tilde{C}_z^{T}Z+\tilde{C}_u^{T}u_0)/\partial \tilde{C}_{zu}} | C^TU=C^Tu]$ is bounded. 

\item\label{VStat}
$E(\abs{Y}|C_u^{T} U=C_u^{T} u, C^{T} U=C^{T} u)$ 
is bounded.

\item\label{AsympNorm}

The function
$E\left[E(Y|V=C_z^T Z+C_u^T u_0)\mid C^T U=C^T u\right]$ 
is $\stwo{m}$th order continuously differentiable
and the function $ E\left\{[ E(Y|V=C_z^T Z + C_u^T u_0)]^2 \mid C^T U=C^T u\right\}
$ is continuous. 
\end{enumerate}
\end{assumption}

\begin{assumption}\label{par}

The kernel bandwidth $\hbar$ of the first step is of the order $n^{-\bar s}$
with $\bar s$ satisfying  $\mbar^{-1}\stwo{m}(2\stwo{m}+\stwo{d})^{-1}<\bar{s}< \min\{(\stwo{m}+\stwo{d})(2\stwo{m}+\stwo{d})^{-1}(1+\dbar)^{-1}, \stwo{m}(2\stwo{m}+\stwo{d})^{-1}\dbar^{-1}\}$;  the kernel bandwidth in the second step is $h=\lambda^{2/(2\stwo{m}+\stwo{d})} n^{-1/(2\stwo{m}+\stwo{d})}$ with $2\stwo{m}>\stwo{d}$ and a constant $\lambda >0$.  
\end{assumption}

Assumptions \ref{UnifConv} and \ref{par} impose some constraints on the orders and bandwidths of the kernels in our two step estimation. 

Assumptions \ref{bound}, \ref{main}\ref{MaZhu1}-\ref{SDRneg} and \ref{par}
are assumed to 
ensure that the estimation errors of SDR estimators $(\hat{C}_{zu}, \hat{C})$ are asymptotically negligible \citep{ma_zhu_2012}. The proof can be found in Lemma 
3 in the Appendix. 
Assumptions \ref{main}\ref{AsympNorm} and \ref{par} are standard for the asymptotic normality of nonparametric kernel estimator \citep{bierens}.

The following result establishes the asymptotic normality as well as the convergence rates of $\hat\psi_{d_1} (u_0)$ and $\hat\psi_{d_2} (u_0)$ defined in (\ref{esti31})-(\ref{esti4}) with a fixed $u_0$. 
\begin{theorem} \label{thmasymp}

(i) If Assumptions \ref{bound}-\ref{par} hold with $d=d_1$, then there exists a function $b(C^T u)$ such that 
\begin{equation}\label{esti3result}\begin{split}
&n^{m/(2m+d_1)} \left[\hat{\psi}_{d_1}(u_0)-\psi(u_0)\right]\Longrightarrow  N\left(\dfrac{\lambda b(C^T u_0)}{f_{C^T U}(C^T u_0)},\dfrac{g(C^T u_0)}{f_{C^T U}(C^T u_0)}\int K^2(t)\diff t \right)
\end{split}\end{equation}
where $\lambda$ is given in the bandwidth (Assumption \ref{par}), $\Longrightarrow$ denotes convergence in distribution,  and 
$g(C^T u)={\rm Var} \left\{ E(Y|V=C_z^T Z + C_u^T u_0)
\mid  C^T U=C^T u\right\}$.\\
 (ii)  If Assumptions \ref{bound}-\ref{par} hold with $d=d_2$ and $C_z$, $C_u$, and $C$  replaced by $D_z$, $BD_u$, and $BD$, respectively, then 
  (\ref{esti3result}) holds with $d_1$, 
  $C_z$, $C_u$, and $C$  replaced by $d_2$, $D_z$, $BD_u$, and $BD$, respectively.   
\end{theorem}

The convergence rate $n^{-m/(2m+d)}$ shown in (\ref{esti3result}) is the optimal convergence rate for $\hat\psi_d(u_0)$, $d=d_1$ or $d_2$. 
Since $\lambda >0$, the asymptotic bias of $\hat\psi_d(u_0)$ has the same order as the asymptotic variance of $\hat\psi_d(u_0)$ and, hence, we should consider asymptotic mean squared error. 
If we choose the bandwidth $h$ to be 
$o(n^{-1/(2m+d)})$, then (\ref{esti3result}) holds with $\lambda$ replaced by 0,
but the convergence rate of the resulting estimator $\hat\psi_d(u_0)$ is slower than 
$n^{-m/(2m+d)}$.

Following  \citet{bierens}, $\hat{\psi}_{d_0}$ defined in (\ref{esti2}) is also  asymptotically normal  with convergence rate $n^{-\stwo{m}/(2\stwo{m}+d_0)}$ if we use the same kernel order $\stwo{m}$ as in the second step of $\hat\psi_{d_1}$ and $\hat\psi_{d_2}$. 
Together with Theorem 1, we conclude that the  convergence rate of  $\hat{\psi}_{d_j}(u_0)$ is $n^{-\stwo{m}/(2\stwo{m}+d_j)}$, $j=0,1,2$, 
 and we can compare  the three estimators by comparing $d_j$'s:  the higher the dimension $d_j$ of the covariate vector used in kernel estimation of the last step, the slower the  convergence rate.

By Theorem 1, the dimension $\dbar$ in $\hat\psi_{d_1}$ and $\hat\psi_{d_2}$ 
does not have a direct influence on the convergence rate. But  some conditions on $\dbar$, $\hbar$, and $\mbar$ in Assumptions \ref{UnifConv} and \ref{par} still need to be satisfied to guarantee the asymptotic normality of the estimator. A high order $\mbar$  may be needed when $\dbar$ is large. Similarly, a high order $\stwo{m}$ may be needed when $d_1$ or $d_2$ is large.


To end this section we provide a discussion on Assumption 1, which requires that 
the density $f_V(v)$ is  bounded away from zero. It is a technical condition, and is sufficient but not necessary, i.e., without Assumption 1, $\hat{\psi}_{d_1}$ and $\hat{\psi}_{d_2}$ may still perform well (see the simulation results in the next section). 
On the other hand, we may use the following transformation method when   Assumption 1 is a concern. Note that 
 $E(Y|V)=E(Y|\varphi (V))$ for any invertible function $\varphi $.
 Thus, we may use $\varphi (V)$ if  the density  $f_{\varphi (V)}$ is bounded away from 0.
 Let $\mu_V$ and $\Sigma_V$ be the mean vector and covariance matrix of $V$, respectively, and $V_S = \Sigma_V^{-1/2} (V-\mu_V)$ be the standardized $V$. 
 We consider transformation $\varphi (V) = ( \Psi_1(V_{S1}), ...,\Psi_{\dbar} (V_{S\dbar}) )^T$, where $V_{Sj}$ is the $j$th component of $V_S$ and $\Psi_j$ is a known distribution function. If $V$ is normally distributed, then a prefect choice of $\Psi_j$ is the cdf of standard 
 normal distribution. Otherwise, we choose $\Psi _j$ to be the empirical distribution based on the $n$ $V_{Sj}$ observations. Since $\mu_V$ and $\Sigma_V$ are unknown, they also have to be estimated using $V$ data. This method is examined in the simulation in Section 4. 

 \section{Simulation Studies}
 
 In this section, we present simulation results to evaluate the performance of $\hat{\psi}_p$, $\hat{\psi}_{d_0}$, $\hat{\psi}_{d_1}$, and $\hat{\psi}_{d_2}$, when the  sample size is $n=200$.
  The estimators are evaluated at $8$ different values of $E(Y|U=u_0)$ with randomly generated $u_0$'s. A second-order Epanechnikov kernel is adopted for all estimators. 
 For the first step of $\hat{\psi}_{d_1}$ and $\hat{\psi}_{d_2}$, bandwidth is denoted as $\hbar$. For estimator $\hat{\psi}_p$ and $\hat{\psi}_{d_0}$ and the second step of $\hat{\psi}_{d_1}$ and $\hat{\psi}_{d_2}$, bandwidths are of the form $h=C n^{-1/( 4+d)}$, where $d$ is $p$, $d_0$, $d_1$, or $d_2$. 
 The values of $\hbar$ and $C$  are selected using 10-fold cross-validation. 
 For SDR, the semiparametric principal Hessian directions method proposed by \citet{ma_zhu_2012} is used.
 
 
 In the first simulation study, we 
 examine the relative performance of different estimators in the ideal situation where we know the column dimensions of matrices $B$, $C$, $D$, $C_{zu}$, and $D_{zu}$.   The situation where these dimensions are unknown is considered in the third simulation study. 
 The following four settings are considered.
 \begin{itemize}
 	\item[(A1)]
 	The model is given by Example 1, where four components of $U$ are uniformly distributed with lower bounds  $-1,-3,-10,$ and $8$, and upper bounds $7,-1,-2$, and $18$;  
 	 three components of $\eta$ are uniformly distributed with lower bounds $0,0$, and 0 and upper bounds $2,3$, and $5$; and  $\epsilon \sim N(0,1)$.
 	\item[(A2)]
 	The model is given by Example 2, where four components of $U$ are uniformly distributed with lower bounds  $3, 0,-5,$ and $8$, and upper bounds $5,9,-2$, and $18$;  
 	three components of $\eta$ are uniformly distributed with lower bounds $0,0$, and 0 and upper bounds $7,3$, and $5$; and  $\epsilon \sim N(0,1)$.
 	 	\item[(A3)]
 	The mode is given by  Example 3, where four components of $U$ are uniformly distributed with lower bounds  $-1,-3,-10,$ and $2$, and upper bounds $3,-1,-2$, and $4$;  
 	three components of $\eta$ are uniformly distributed with lower bounds $0,0$, and 0 and upper bounds $7,3$, and $5$; and  $\epsilon \sim N(0,1)$.
 \end{itemize}
 
 We may replace $\hat{B}$ in $\hat\psi_{d_0}$ 
 by the proposed SDR method in \citet{hung_liu_lu_2015} 
 that  first constructs a $Z$-envelope $ \supseteq {\cal S}_{Y|U} = \SS(B)$ and 
 then estimates $B$ in the $Z$-envelope. Let $\widetilde{\psi}_{d_0}$ denote 
 the resulting estimator of $\psi(u_0)$.  Although $\widetilde{\psi}_{d_0}$
 does not improve 
 $\hat{\psi}_{d_0}$ in terms of convergence rate, $\widetilde{\psi}_{d_0}$ may have a better finite sample performance than $\hat{\psi}_{d_0}$ because the auxiliary $Z$ information is used in estimating $B$ through the $Z$-envelope. However, it can be shown that in settings (A1)-(A3), the $Z$-envelope is the whole space ${\cal R}^4$ and, hence, $\widetilde{\psi}_{d_0} = \hat{\psi}_{d_0}$ as discussed by 
 \citet{hung_liu_lu_2015}. To see whether  $\widetilde{\psi}_{d_0} $ improves $\hat{\psi}_{d_0}$, we consider another setting as follows:
 \begin{itemize}
 	\item[(A1$'$)]
 	The model, $Y$, and $U$ are the same as those in (A1), but $Z=z_1=\abs{u_1-u_2}+\eta_1$ is one dimensional.  
 	\end{itemize}
   Under setting (A1$'$), the $Z$-envelope is 3-dimensional and 
 \[ {\cal S}_{Y|U}^Z = \left( \begin{array}{rr}
 0 & 0  \\
 0 & 0  \\
 1 & 0  \\
 0 & 1  \end{array} \right) \quad \quad
 {\cal S}_{Z|U} = \left(  \begin{array}{rr} 
 1  \\
 -1 \\
 0 \\
 0  \end{array} \right) \quad\quad
 Z\text{-envelope} =  \left(  \begin{array}{rrr} 
 1 & 0 & 0  \\
 -1 & 0 & 0 \\
 0 & 1 & 0\\
 0 & 0 & 1\end{array} \right) \quad\quad \]

 \begin{table}[h]
 	\caption{Absolute value of bias and  root mean squared error (RMSE) of different estimators at  eight values of  $E(Y|U=u_0)$ based on $n=200$ and 1,000 simulations.} \label{A200}
 	\centering
	\scalebox{0.92}{
 	\begin{tabular}{cccrrrrrrrr}
 		\hline
 		&& &\multicolumn{8}{c}{value of $E(Y|U=u_0)$} \\ 
 		setting &	quantity & method \hspace{3mm}  &  1.42 & 1.71 & 1.79 & 1.83 & 1.90 & 4.31 & 7.88 & 8.61 \\ 
 		\hline
 		(A1) &	|bias| & $\hat\psi_p$  & 0.12 & 0.31 & 0.11 & 0.14 & 0.11 & 0.86 & 0.88 & 0.45 \\ 
 		&& $\hat\psi_{d_0}$ &  0.04 & 0.22 & 0.00 & 0.05 & 0.01 & 0.45 & 0.80 & 0.16 \\ 
 		&& $\hat\psi_{d_1}$ & 0.00 & 0.15 & 0.11 & 0.01 & 0.10 & 0.09 & 0.56 & 0.02 \\ 
 		&& $\hat\psi_{d_2}$& 0.00 & 0.18 & 0.12 & 0.00 & 0.11 & 0.10 & 0.63 & 0.03 \\ \cline{4-11}
 & RMSE &  $\hat\psi_p$ & 0.36 & 0.68 & 0.44 & 0.47 & 0.41 & 1.13 & 1.84 & 1.34 \\ 
 		&& $\hat\psi_{d_0}$ &  0.33 & 0.64 & 0.44 & 0.44 & 0.40 & 0.72 & 1.51 & 1.17 \\ 
 		&& $\hat\psi_{d_1}$  & 0.22 & 0.33 & 0.26 & 0.25 & 0.24 & 0.30 & 0.93 & 0.54 \\ 
 		&& $\hat\psi_{d_2}$ & 0.25 & 0.34 & 0.28 & 0.29 & 0.26 & 0.28 & 0.90 & 0.55 \\ \cline{4-11}
 (A1$'$) &	|bias| & $\hat\psi_p$  & 0.13 & 0.28 & 0.11 & 0.15 & 0.11 & 0.88 & 0.88 & 0.46 \\ 
 	& & $\hat\psi_{d_0}$  & 0.04 & 0.21 & 0.00 & 0.04 & 0.01 & 0.45 & 0.87 & 0.14 \\ 
 	& &$\widetilde\psi_{d_0}$ & 0.03 & 0.23 & 0.02 & 0.03 & 0.01 & 0.43 & 0.75 & 0.08 \\ 
 	& &$\hat\psi_{d_1}$   & 0.02 & 0.14 & 0.09 & 0.02 & 0.08 & 0.13 & 0.59 & 0.07 \\ 
 	& &$\hat\psi_{d_2}$  & 0.00 & 0.15 & 0.08 & 0.01 & 0.06 & 0.13 & 0.59 & 0.15 \\ \cline{4-11}
 &RMSE & $\hat\psi_p$ & 0.36 & 0.66 & 0.45 & 0.46 & 0.42 & 1.14 & 1.84 & 1.34 \\ 
 	&& $\hat\psi_{d_0}$ &  0.34 & 0.63 & 0.41 & 0.44 & 0.39 & 0.75 & 1.63 & 1.17 \\ 
 	& &$\widetilde\psi_{d_0}$  &  0.34 & 0.61 & 0.39 & 0.48 & 0.39 & 0.74 & 1.54 & 1.22 \\ 
 	& &$\hat\psi_{d_1}$    & 0.24 & 0.40 & 0.29 & 0.31 & 0.28 & 0.34 & 0.95 & 0.60 \\ 
 	& &$\hat\psi_{d_2}$   & 0.26 & 0.42 & 0.32 & 0.33 & 0.30 & 0.36 & 0.95 & 0.67 \\ 
 	\cline{4-11}
		&& &\multicolumn{8}{c}{value of $E(Y|U=u_0)$} \\ 
	 &	 & & 23.57 & 26.83 & 28.40 & 29.74 & 35.00 & 36.91 & 39.46 & 41.47 \\ 
\cline{4-11}
(A2) &	|bias| & $\hat\psi_p$ & 0.25 & 0.67 & 1.19 & 0.03 & 1.28 & 2.15 & 3.03 & 3.68 \\ 
&& $\hat\psi_{d_0}$ & 0.57 & 0.00 & 0.75 & 0.31 & 0.25 & 0.30 & 1.34 & 1.99 \\ 
&& $\hat\psi_{d_1}$  & 1.30 & 1.24 & 2.43  & 2.32 & 2.65 & 2.77 & 4.38 & 5.01 \\ 
&& $\hat\psi_{d_2}$ & 0.25 & 0.23 & 0.65  & 0.50 & 0.63 & 0.82 & 1.49 & 2.23 \\ \cline{4-11}
& RMSE &  $\hat\psi_p$  & 3.67 & 4.00 & 4.12 & 4.63 & 3.57 & 3.50 & 4.81 & 5.43 \\ 
&& $\hat\psi_{d_0}$   & 3.28 & 3.42 & 3.61 &  3.88 & 3.03 & 2.64 & 3.51 & 4.03 \\ 
&& $\hat\psi_{d_1}$  & 3.97 & 4.03 & 4.96 & 5.30 & 4.33 & 4.08 & 5.92 & 6.53 \\ 
&& $\hat\psi_{d_2}$ & 2.96 & 3.09 & 3.14  & 3.38 & 2.81 & 2.58 & 3.28 & 3.66 \\ 
\cline{4-11}
&& &\multicolumn{8}{c}{value of $E(Y|U=u_0)$} \\ 
&	 &    & 28.22 & 28.91 & 67.44 & 82.46 & 86.44 & 118.5 & 137.7 & 138.8 \\ 
\cline{4-11}
(A3) &	|bias| & $\hat\psi_p$ & 2.71 & 2.14 & 1.38 & 6.42 & 1.30 & 12.50 & 14.11 & 15.85 \\ 
&& $\hat\psi_{d_0}$ & 1.17 & 0.74 & 0.35 & 2.85 & 0.48 & 5.99 & 7.34 & 8.43 \\ 
&& $\hat\psi_{d_1}$  & 1.17 & 0.44 & 0.30 & 2.84 & 0.45 & 5.80 & 7.72 & 9.24 \\ 
&&$\hat\psi_{d_2}$ & 1.60 & 1.02 & 0.22 & 3.40 & 0.81 & 7.60 & 9.83 & 10.67 \\ \cline{4-11}
&RMSE &  $\hat\psi_p$  & 4.68 & 4.81 & 5.42 & 9.75 & 7.60 & 14.97 & 15.82 & 17.83 \\ 
&& $\hat\psi_{d_0}$	 & 2.14 & 2.21 & 2.69 & 7.83 & 3.16 & 8.16 & 9.49 & 10.72 \\ 
&& $\hat\psi_{d_1}$  & 1.83 & 1.67 & 1.82 & 6.02 & 2.67 & 7.88 & 9.59 & 12.01 \\ 
&&$\hat\psi_{d_2}$ & 2.46 & 2.24 & 2.41 & 8.95 & 2.99 & 10.06 & 12.15 & 13.07 \\ 
\hline
 	\end{tabular}}
 \end{table}
 
 Simulation results with 1,000 runs under these four settings are reported in 
 Table \ref{A200}, which contains the 
 absolute value of bias and root-mean-squared error (RMSE) of different estimators.
 It can be seen from  Table \ref{A200}  that the performance of various estimators supports  our theory.
 In summary, the results in Table \ref{A200} indicate that,  in terms of  RMSE, (i) $\hat\psi_{d_0}$ is better than $\hat\psi_p$ due to dimension 
 reduction from $p $ to $d_0$; (ii) in setting (A1), $d_1 <d_0$ and $d_2 < d_0$, and our proposed estimators $\hat\psi_{d_1}$ and $\hat\psi_{d_2}$ are better than $\hat\psi_{d_0}$; (iii) in setting (A2), 
 $d_2 < d_0 < d_1$, and $\hat\psi_{d_2}$ is better than $\hat\psi_{d_0}$ but $\hat\psi_{d_1}$ is worse than 
 $\hat\psi_{d_0}$; (iv) in setting (A3), $d_1 < d_2 =d_0$, and $\hat\psi_{d_1}$ is better than $\hat\psi_{d_0}$ and 
 $\hat\psi_{d_2}$ is comparable to $\hat\psi_{d_0}$ except for the cases where 
 $E(Y|U=u_0)=137.7$ and 138.8; (v) in setting (A1$'$), $\widetilde{\psi}_{d_0}$ is only slightly better than $\hat\psi_{d_0}$ for some cases and is still worse than $\hat\psi_{d_1}$ or $\hat\psi_{d_2}$.

In the second simulation study we would like to examine the effect of covariate densities not bounded away from 0 and the use of transformation  discussed in the end of Section 3. We  consider the following setting:
 \begin{itemize}
	\item[(B)]
	$Y = 2z_3(z_1+u_3)+\epsilon=2(u_1+\eta_3)(u_1+u_3+\eta_1)+\epsilon$, 
 $z_1 = u_1 + \eta_1$, 
$			z_2=u_2+\eta_2$, 
$			z_3=u_1+\eta_3 $,  
	 $U\sim N(\mu,\Sigma)$, $\eta\sim N(0,\Sigma_{\eta})$, and $\epsilon\sim N(0,3)$, where $\mu=(3,-2,-6,3)^T$, 
	\[
	\Sigma=\left( \begin{array}{rrrr} 1 & 0 & 0.2 & 0\\0 & 0.3 & 0 & 0 \\ 0.2 & 0 & 2 & 0 \\ 0 & 0 & 0 & 1 \end{array} \right) \quad\quad \mbox{and} \quad\quad
	\Sigma_{\eta}=\left( \begin{array}{rrrr}3 &0 & 0\\ 0 & 1 & 0 \\0 & 0 & 2\end{array} \right)
	\]
\end{itemize}
Under setting (B), 
\[ B = \left( \begin{array}{rr}
1 & 0  \\
0 & 0  \\
0 & 1  \\
0 & 0  \end{array} \right) \quad \quad
C_z = \left(  \begin{array}{rr} 
1 & 0 \\
0 & 0 \\
0 & 1 \end{array} \right) \quad\quad
C_u =  \left(  \begin{array}{rr} 
0 &  0  \\
0 & 0  \\
1 & 0 \\
0 & 0 \end{array} \right) \quad\quad
C= \left(  \begin{array}{r}	
1  \\ 
0  \\
0  \\
0  \end{array} \right) \]
\[ D_z = \left(  \begin{array}{rr} 
1  & 0  \\
0 & 0 \\
0 & 1 \end{array} \right) \quad\quad
D_u =  \left(  \begin{array}{rrrr} 
0 & 0 \\
1 & 0  \end{array} \right) \quad\quad
D = 
\left(  \begin{array}{rrr}	
1 \\
0 \end{array} \right) \quad\quad 
BD =  \left(  \begin{array}{rrr}	
1 \\ 
0 \\
0 \\
0 \end{array} \right) \]
Note that $d_1=d_2=1< d_0=2 < p=4$. We still assume that 
the dimensions of these matrices are  known. 
We apply the transformation discussed in the end of Section 3 with 
$\Psi_j$ being either the standard normal or the empirical distribution.  The resulting estimators are denoted by 
$\hat\psi_{d}^N$ and $\hat\psi_{d}^E$, respectively, $d=d_1$ or $d_2$.

 \begin{table}[h]
 	\caption{Absolute value of bias and  root mean squared error (RMSE) of different estimators at  eight values of  $E(Y|U=u_0)$ based on $n=200$ and 1,000 simulations.}\label{B200}
 	\centering
	\scalebox{0.92}{
	\begin{tabular}{cccrrrrrrrr}
	\hline
	&& &\multicolumn{8}{c}{value of $E(Y|U=u_0)$} \\ 
setting &	quantity & method \hspace{3mm}   & -23.93  & -22.70  & -18.53 & -16.99 & -14.32 & -12.35 & -10.31& -9.09  \\ 
\hline 
(B) &	$|$bias$|$ & $\hat\psi_p$  & 3.39 & 2.83  & 1.56 & 0.01  & 1.32  & 2.32 & 2.73 & 3.06 \\ 
&& $\hat\psi_{d_0}$  & 1.80 & 2.44 & 2.57 & 0.35  & 0.66  & 1.57 & 1.94& 2.04  \\  
&& $\hat\psi_{d_1}$  & 1.11 & 1.71  & 1.63 & 0.58 & 0.09 & 0.74  & 1.24 & 1.36\\ 
&&$\hat\psi_{d_1}^N$  & 0.79 & 1.48 & 1.55 & 0.55  & 0.22 & 0.45  & 0.91  & 0.99 \\  
&	& $\hat\psi_{d_1}^E$ & 0.83 & 1.50 & 1.56 & 0.61  & 0.26 & 0.47  & 0.90  & 0.98 \\ 
&& $\hat\psi_{d_2}$ & 1.10  & 2.15  & 2.50 & 0.71  & 0.02 & 0.78& 1.17  & 1.23 \\ 
&&$\hat\psi_{d_2}^N$ & 0.82  & 1.95 & 2.49 & 0.65  & 0.08 & 0.55 & 0.91 & 0.87 \\ 
&	& $\hat\psi_{d_2}^E$  & 0.81 & 1.96  & 2.49 & 0.67  & 0.11 & 0.58 & 0.90  & 0.86  \\ 
\cline{4-11}
&	RMSE & $\hat\psi_p$  & 4.15 & 3.69  & 6.72 & 3.55  & 3.20 & 3.41  & 3.42  & 3.63 \\ 
&& $\hat\psi_{d_0}$  & 3.26 & 3.54  & 4.59 & 3.27 & 3.07 & 3.03 & 2.89& 2.91\\ 
&& $\hat\psi_{d_1}$  & 2.14 & 2.50 & 2.67 & 2.25 & 2.21 & 2.15  & 2.02 & 2.03 \\
&&$\hat\psi_{d_1}^N$  & 2.03 & 2.46 & 2.69 & 2.51 & 2.47  & 2.28  & 1.95 & 1.90  \\ 
&	& $\hat\psi_{d_1}^E$  & 2.06 & 2.47  & 2.68 & 2.51  & 2.45  & 2.26 & 1.94  & 1.88 \\ 
&& $\hat\psi_{d_2}$  & 2.51  & 3.05 & 4.09 & 2.91 & 2.73  & 2.47  & 2.12 & 2.09 \\ 
&&$\hat\psi_{d_2}^N$  & 2.48  & 3.01  & 4.16 & 3.04  & 2.90 & 2.58  & 2.06 & 1.93\\ 
&	& $\hat\psi_{d_2}^E$  & 2.48  & 3.02 & 4.25 & 3.08  & 2.96 & 2.60  & 2.06 & 1.96 \\  
\hline
\end{tabular}}
\end{table}
 
 From the results in Table \ref{B200}, $\hat{\psi}_{d_2}$ has larger RMSE than $\hat{\psi}_{d_1}$, which may be true in general 
 when $d_1=d_2=1$ because when   $\hat{\psi}_{d_1}$ and  $\hat{\psi}_{d_2}$ have the same convergence rate, 
  $\hat{\psi}_{d_2}$ may have worse finite sample performance as it 
requires an additional application of SDR. Even this is the case, $\hat{\psi}_{d_2}$ still outperforms $\hat{\psi}_{d_0}$ and  $\hat{\psi}_{p}$.
 The difference between  $\hat{\psi}_{d}^N$ and  $\hat{\psi}_{d}^E$ based on two transformation methods is  small, indicating that the use of empirical distribution is adequate. Since $\hat{\psi}_{d}^N$ and  $\hat{\psi}_{d}^E$ are comparable with 
  $\hat{\psi}_{d}$, the estimator without covariate transformation, the results show that Assumption 1 is not necessary for better performance of   $\hat{\psi}_{d}$ over 
  $\hat{\psi}_{d_0}$ and  $\hat{\psi}_{p}$.

 So far the dimensions
 of matrices $B$, $C_{zu}$, $C$, $D_{zu}$, and $D$  are assumed known. 
 In the third simulation study, we estimate these dimensions  using a bootstrap procedure  described by \citet{dong_li_2010} and recommended by \citet{ma_zhu_2012}, with bootstrap Monte Carlo size $30$. 
 Settings (A1), (A2), and  (B) are considered. Under setting (B), only results with 
$\hat{\psi}_{d}^N$ are reported since $\hat{\psi}_d$ and 
 $\hat{\psi}_{d}^N$ are similar.  Simulation results are shown in Table 3.

It can be seen from Table 3 that, in terms of RMSE,  $\hat{\psi}_{d_1}$ and  $\hat{\psi}_{d_2}$ perform worse than those obtained using the true dimensions of matrices (Tables 1-2). But they are still better than $\hat{\psi}_{d_0}$ except for $\hat\psi_{d_2}$ in three cases under setting (B). Note that estimation of dimensions of subspaces is a difficult topic in the literature of SDR. More accurate estimators of the dimensions of matrices in SDR will result in better performance of our estimators 
 $\hat{\psi}_{d_1}$ and  $\hat{\psi}_{d_2}$.

 \begin{table}[h]
 	\caption{Absolute value of bias  and  root mean squared error (RMSE) of different estimators at  eight values of  $E(Y|U=u_0)$ based on $n=200$ and 1,000 simulations; the dimensions of matrices in SDR are selected by bootstrap.}\label{bootA200}
\centering
	\scalebox{0.92}{
 	\begin{tabular}{cccrrrrrrrr}
 		\hline
 		&& &\multicolumn{8}{c}{value of $E(Y|U=u_0)$} \\ 
 		setting &	quantity & method \hspace{3mm}   & 1.42 & 1.71 & 1.79 & 1.83 & 1.90 & 4.31 & 7.88 & 8.61 \\ 
 		\hline
 		(A1) &	|bias| & $\hat\psi_p$  & 0.13 & 0.31 & 0.10 & 0.15 & 0.11 & 0.86 & 0.86 & 0.46 \\ 
         & & $\hat\psi_{d_0}$  & 0.01 & 0.20 & 0.06 & 0.00 & 0.05 & 0.24 & 0.70 & 0.22 \\ 
         &	 & $\hat\psi_{d_1}$& 0.03 & 0.18 & 0.09 & 0.02 & 0.08 & 0.25 & 0.51 & 0.14 \\ 
         &	 & $\hat\psi_{d_2}$& 0.01 & 0.20 & 0.12 & 0.01 & 0.10 & 0.15 & 0.60 & 0.01 \\ \cline{4-11}	
\cline{4-11}	
         & RMSE & $\hat\psi_p$ & 0.36 & 0.68 & 0.45 & 0.47 & 0.41 & 1.13 & 1.82 & 1.34 \\ 
         & & $\hat\psi_{d_0}$ & 0.26 & 0.49 & 0.32 & 0.33 & 0.29 & 0.46 & 1.11 & 0.87 \\ 
         &	& $\hat\psi_{d_1}$ & 0.22 & 0.38 & 0.25 & 0.26 & 0.24 & 0.47 & 0.93 & 0.66 \\ 
         & & $\hat\psi_{d_2}$ & 0.20 & 0.33 & 0.25 & 0.24 & 0.24 & 0.34 & 0.89 & 0.56 \\ 
 		\cline{4-11}
 		&& &\multicolumn{8}{c}{value of $E(Y|U=u_0)$} \\ 
 	&	 &    & 23.57 & 26.83 & 28.40 & 29.74 & 35.00 & 36.91 & 39.46 & 41.47 \\ 
 		\cline{4-11}
       (A2) &	|bias| & $\hat\psi_p$   & 0.30 & 0.59 & 1.23 & 0.09 & 1.30 & 2.06 & 3.08 & 3.56 \\ 
        & & $\hat\psi_{d_0}$& 0.53 & 0.03 & 0.77 & 0.34 & 0.41 & 0.40 & 1.58 & 2.00 \\ 
        &	 & $\hat\psi_{d_1}$ & 0.20 & 0.10 & 0.93 & 0.86 & 1.45 & 1.64 & 2.99 & 3.51 \\ 
        &	 & $\hat\psi_{d_2}$& 0.14 & 0.21 & 0.67 & 0.41 & 0.75 & 0.96 & 1.77 & 2.41 \\ \cline{4-11}
  & RMSE & $\hat\psi_p$ & 3.59 & 4.00 & 4.06 & 4.68 & 3.66 & 3.45 & 4.94 & 5.38 \\ 
        & & $\hat\psi_{d_0}$& 3.29 & 3.48 & 3.58 & 4.03 & 3.08 & 2.83 & 3.79 & 4.07 \\ 
        &	& $\hat\psi_{d_1}$& 3.14 & 3.45 & 3.46 & 4.08 & 3.31 & 3.23 & 4.45 & 4.88 \\ 
        & & $\hat\psi_{d_2}$  & 2.88 & 3.11 & 3.12 & 3.35 & 2.83 & 2.62 & 3.38 & 3.82 \\ 
   		  \cline{4-11}
 		&& &\multicolumn{8}{c}{value of $E(Y|U=u_0)$} \\ 
  &	 &   &  -23.93 & -22.70& -18.53 & -16.99 & -14.32  & -12.35 & -10.31  & -9.09  \\ 
 \cline{4-11} 
 (B) &	$|$bias$|$ & $\hat\psi_p$ & 3.46 & 2.87 & 1.41 & 0.01  & 1.36  & 2.34 & 2.77 & 3.11  \\ 
 && $\hat\psi_{d_0}$ & 2.56  & 2.66 & 2.50 & 0.40 & 0.85 & 1.67  & 1.72 & 1.77\\ 
 && $\hat\psi^N_{d_1}$  & 2.61  & 1.99 & 2.21 & 0.86 & 0.22  & 0.61  & 0.11 & 0.22 \\ 
 && $\hat\psi^N_{d_2}$ & 2.20 & 2.66 & 3.13 & 0.95 & 0.01 & 0.55  & 0.19 & 0.07 \\ \cline{4-11}
 &	RMSE  & $\hat\psi_p$ & 4.17  & 3.66 & 6.29 & 3.42 & 3.16  & 3.36  & 3.43  & 3.64 \\ 
 && $\hat\psi_{d_0}$ & 3.73 & 3.58 & 4.37 & 2.94 & 2.88   & 2.97   & 2.79 & 2.86 \\ 
 && $\hat\psi^N_{d_1}$ & 3.50  & 2.84  & 3.69 & 2.75 & 2.52 & 2.27  & 2.08 & 2.18 \\ 
 && $\hat\psi^N_{d_2}$ & 3.49  & 3.62 & 4.67 & 3.11 & 2.86  & 2.59   & 2.16 & 2.25 \\ 
 \hline
	\end{tabular}}
 \end{table}

  \section{Data Analysis}
Breast cancer has been taking a toll on the lives of women. The good news is that regular mammography screening can help to reduce mortality. \citet{champion_et_al_2014} did a Computer and Phone (CAPE) study including two tailored intervention methods, mailed DVD (abbreviated as DVD) and telephone counseling (abbreviated as TC). A CAPE randomized controlled trial was
conducted to determine whether the two interventions were more efficacious than the usual care (abbreviated as UC) method  at promoting  mammography screening  among women who are non-adherent to breast cancer screening guidelines at baseline.
If the answer is yes, then we are further interested in 
which of DVD, TC, and UC methods 
is more efficacious at promoting  mammography screening for women with 
a particular set of demographic values. This involves estimation of 
$ E(Y|U=u_0, a)$ for fixed $u_0$ and $a$ as well as 
$\mu (a) = E(Y|a) = E\{E(Y|U,a)\}$, where $Y$ is an outcome of interest, $U$ is a vector of demographic variables, and $a =$ DVD, TC, and  UC corresponding to  mailed DVD, telephone counseling, and usual care, respectively, which is treated as a treatment indicator. 

In the CAPE dataset, there are 26 demographic variables such as  
age, years of education, and household income, collected at baseline of the study. 
The outcome we consider is perceived barriers, one of the  health belief variables related with health behaviors according to the Health Belief Model. The variable of 
perceived barriers is the sum of grades (typically 1-5) to  questions such as ``I am afraid of finding out that I might have breast cancer'', ``the treatment for breast cancer is worse than the cancer itself'', ``having a mammogram is painful for me'', ``I don't have the time to get a mammogram'', etc. Other 
belief variables include perceived risk, perceived benefits,  self-efficacy, breast cancer fear, and fatalism. We focus on  perceived barriers for illustration. 

The outcome of perceived barriers in the CAPE dataset is actually longitudinal and observed at baseline, one month after baseline, and six months after baseline. 
We are interested in the outcome of perceived barriers after six months from the time an individual is assigned to one of DVD, TC, and UC. Thus,  $Y=bar3tot$ is the score of perceived barriers at  month six after baseline and the scores of perceived barriers at
baseline ($bar1tot$) and at month one after baseline  ($bar2tot$) are two components of $Z$ that are closely related with $Y$ but not available in the future prediction. 

After eliminating units with missing data, 
the training dataset for our analysis contains  $357$, $434$, and $423$ sampled units 
for the DVD, TC, and UC methods, respectively. 

Note that 26 demographic variables (covariates) are too many even for SDR. 
Thus, we follow the idea in \citet{mai_zou_2015}  that applies fused Kolmogorov filter  to screen out some  demographic variables  not useful  in predicting $Y$. A fused Kolmogorov filter statistic that measures the dependence between a certain covariate $X_j$ and the continuous response variable $Y$ is defined as
$$\hat K_j = \sum_{i=1}^N \hat{K}_j^{\bm{G_i}} $$
 with $\hat{K}_j^{\bm{ G_i}} = \max_{l,m}\sup_x\abs{\hat{F}_j(x|H^i=l)-\hat{F}_j(x|H^i=m)}$. $\bm{G_i}$ is a uniform partition of $Y$ with $G_i$ slices containing
the intervals bounded by the $l/G_i$th sample quantiles of $Y$ for $l = 0, . . . ,G_i$, and $H^i=l$ if $Y$ is in the $l$th slice. $\hat{F}_j(x|H^i)$ is the empirical CDF of $X_j$ conditional on $H^i$. $N$ is the total number of different partitions. We pick uniform partitions $\bm{G_i}$s with $G_i=3,4 ,... ,[\log n]$ and calculate fused Kolmogorov filter statistics for each variable $X_j$ in $bar1tot$, $bar2tot$ and all demographic variables under DVD, TC and UC methods. Results are presented in Figure \ref{filter}. A higher fused Kolmogorov filter statistic indicates a stronger relationship between a covariate and $Y$. It coincides with our instinct that the two components of $Z$, $bar1tot$ and $bar2tot$, are the best predictors of $Y$ in all the three sub-datasets, as their fused Kolmogorov filter statistics are much greater than those for the demographic variables. 
For the sub-dataset under DVD,  we keep 5 demographic variables with the biggest  fused Kolmogorov filter statistics next to $bar1tot$ and $bar2tot$ and treat them as $U$, since  there is a sudden decrease in fused Kolmogorov filter statistics at the $6$th demographic variable SF12RP1.    As suggested in Figure \ref{filter}, variable $U$ under DVD contains   ``income3'', ``educyrs'', ``yearmamsum'', ``SF12GH1'' and ``age'', which represent ``household income'', ``years of education'', ``number of years had a mammogram in the past 2 to 5 years'', ``SF12 general health scale score'' and ``age'', respectively. For the other two sub-datasets under TC and UC, for simplicity we just keep the 5 demographic variables next to $bar1tot$ and $bar2tot$, although these variables may be different from those under DVD.  The selected $U$ variables ``hcreminder'', ``SF12VT1 '' and ``SF12MH1'' under TC or UC represent ``whether or not received any reminders from your health care facility that it was time for you to have a mammogram'', ``SF12 vitality scale score'' and ``SF12 mental health scale score''.


\begin{figure}[h]
\begin{center}
\caption{Fused Kolmogorov Filter Statistics for all demographic variables under DVD, TC, and UC methods}\label{filter}
\begin{tikzpicture}
  \node (img1)  {\includegraphics[scale=0.26]{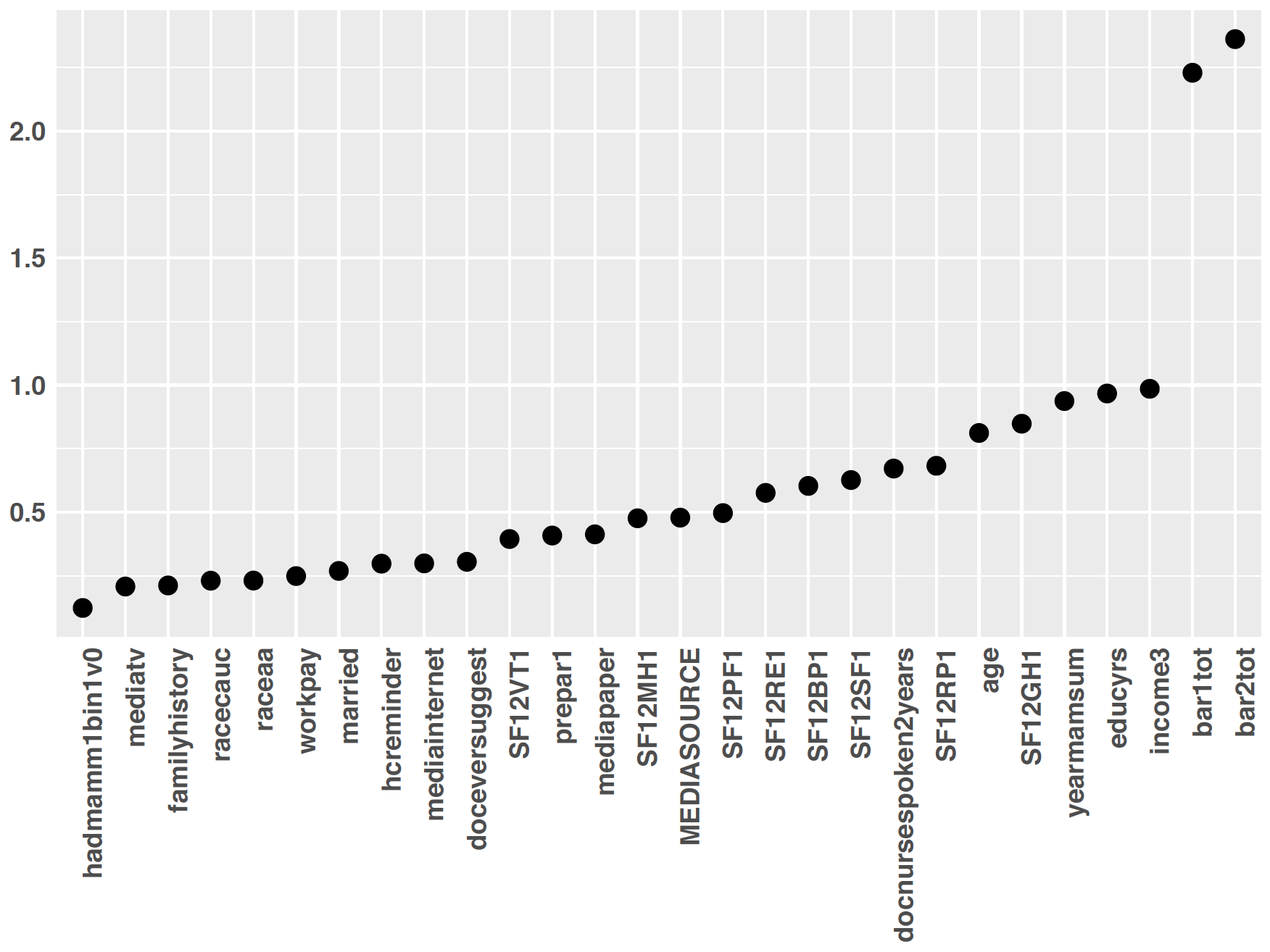}};
  \node[left=of img1, node distance=0cm, anchor=center,yshift=-0.7cm] {\footnotesize DVD};
  \node[below=of img1,yshift=0.5cm] (img2)  {\includegraphics[scale=0.26]{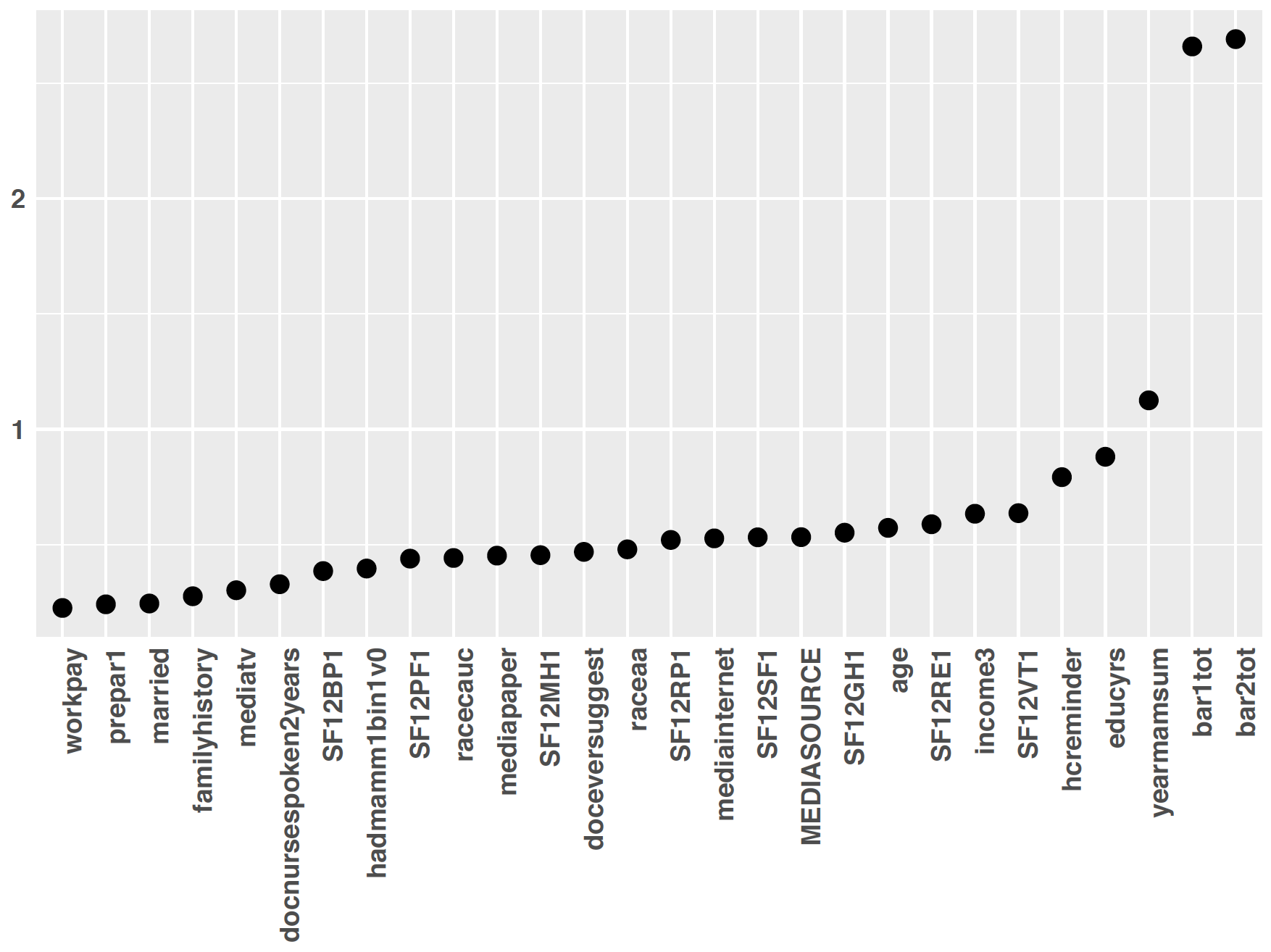}};
  \node[left=of img2, node distance=0cm, anchor=center,yshift=-0.7cm] {\footnotesize TC};
    \node[below=of img2,yshift=0.5cm] (img3)  {\includegraphics[scale=0.26]{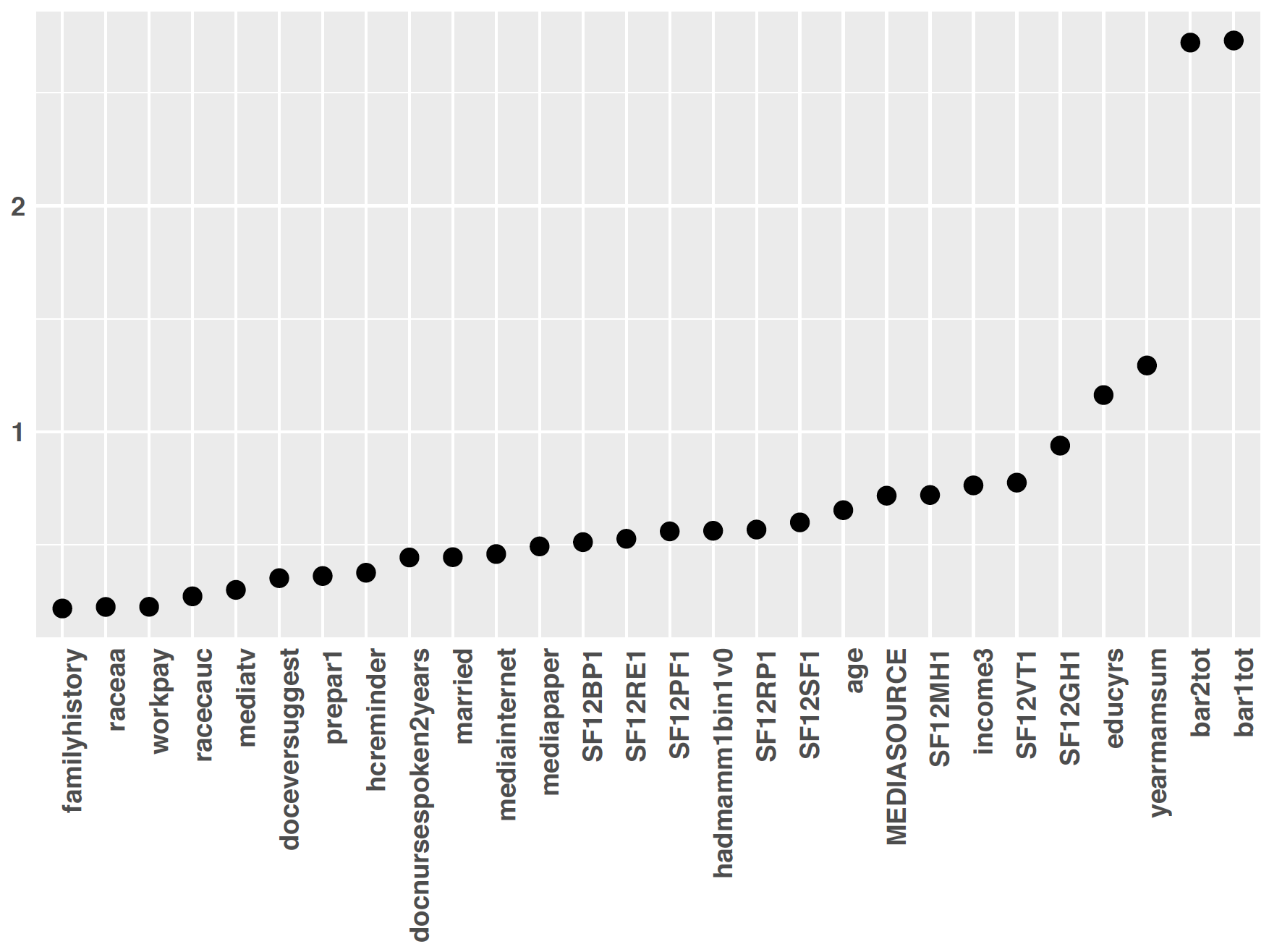}};
  \node[below=of img3, node distance=0cm, yshift=1cm] {\small demographic variable};
  \node[left=of img3, node distance=0cm, anchor=center,yshift=-0.7cm] {\footnotesize UC};
\end{tikzpicture}
\end{center}
\end{figure}

First, we would like to examine whether DVD and TC are more efficacious than UC   at promoting  mammography screening. Note that this can be done using $Y$ data only, i.e., the two sample t-tests based on sample means and variances. The results from the two-sample t-tests, however, show that there is no significant difference among the three methods, i.e., the p-values for rejecting $\mu({\rm DVD}) = \mu ({\rm UC})$,   
$\mu({\rm TC})  = \mu ({\rm UC})$, and $\mu({\rm DVD}) = \mu ({\rm TC})$ are 
0.87, 0.33, and 0.43, respectively. The insignificance results may be due to large variability in $Y$ data. If we make use of covariates, the results may be different. 

Under each DVD, TC, and UC, we compute estimators  $\hat{\psi}_p$, $\hat{\psi}_{d_0}$,  $\hat{\psi}_{d_1}$, and $\hat{\psi}_{d_2}$ without covariate transformation, using the procedures given in Section 2.  
Dimensions of matrices for using SDR  estimated by the bootstrap as described  in the simulation are given as follows.
\begin{table}[H]
	\centering
	\begin{tabular}{cccccc}
		\hline
		& \multicolumn{5}{c}{estimated dimension} \\
	group 	&  $d_0$ & $\dbar_1$ & $d_1$ & $\dbar_2$ & $d_2$ \\ 
		\hline
		DVD  & 2 & 2 & 2 & 2 & 1\\ 
		TC  & 1 & 2 & 1 & 1 & 1\\
		UC  & 1 & 2 & 2 & 2 & 1\\
		\hline
	\end{tabular}
\end{table}

Boxplots of values of $Y_i$ and $\hat\psi_{d}(U_i) $ with $d=p, d_0, d_1$, or $d_2$ are shown in Figure 2. 
It can be seen that clearly $Y_i$ without using any covariate has much larger variability than $\hat\psi_{d}(U_i) $'s using covariate information, and 
$\hat\psi_{p}(U_i) $ has the largest
variability among the four estimators using covariates.   For DVD group, 
 $\hat\psi_{d_2}$ has the least variability but 
$\hat\psi_{d_0}$ is not too bad. For TC, $\hat\psi_{d_1}$ is the best. For UC, 
$\hat\psi_{d_1}$ and $\hat\psi_{d_2}$ are comparable and are much less variable than $\hat\psi_{d_0}$.

Note that we can estimate $\mu(a) = E\{ E(Y|U,a)\}$ using the average of
$\hat\psi_{d}(U_i)  $'s with $U_i$'s in each method group and  $d=p, d_0, d_1$, or $d_2$.
 Using $\hat\psi_{d}(U_i)  $'s and 10,000 random permutations, we obtain  p-values 
 for testing various hypotheses based on  
 $d=p, d_0, d_1$, or $d_2$.  The results are shown in Table 4. The reason we also consider one-sided tests is because the method with smaller $\mu(a)$ is better at promoting  mammography screening.
 
	\begin{figure}[h]
		\centering
	\caption{Boxplots of $Y_i$ and $\hat\psi_d(U_i)$ in each method group, $d=p, d_0, d_1$, or $d_2$.}
		\includegraphics[width=0.8\textwidth]{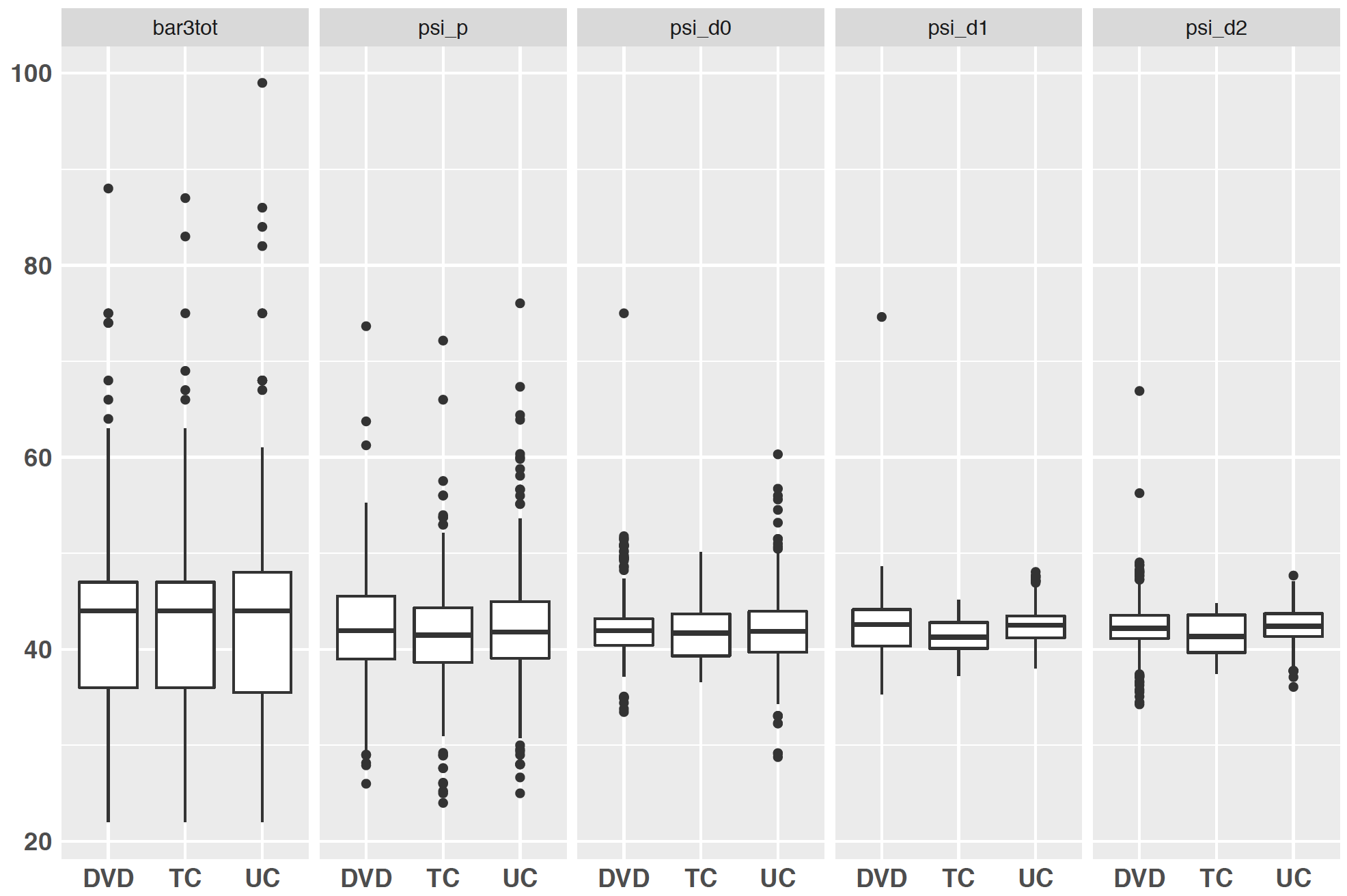}
	\end{figure}

		\begin{table}[h]
			\centering
			\caption{p-values using $\hat\psi_d$ and 10,000 permutations under different hypotheses}
			\begin{tabular}{ccccccc}
				\hline
				&&&\multicolumn{4}{c}{method} \\
			$H_0$	& $H_1$ & & $\hat{\psi}_{p}$ & $\hat{\psi}_{d_0}$ & $\hat{\psi}_{d_1}$ & $\hat{\psi}_{d_2}$ \\ 
				\hline
				$\mu({\rm DVD}) = \mu({\rm UC})$  & $\mu({\rm DVD}) \neq \mu({\rm UC})$  & & 0.55 & 0.65 & 0.49 & 0.88 \\ 
$\mu({\rm DVD}) \geq \mu({\rm UC})$  & $\mu({\rm DVD}) < \mu({\rm UC})$  & & 0.27 & 0.68 & 0.25 & 0.56  \\ 
$\mu({\rm TC}) = \mu({\rm UC})$  & $\mu({\rm TC}) \neq \mu({\rm UC})$  & 	&  0.04 &  0.01 &  0.00 &  0.00 \\ 
$\mu({\rm TC}) \geq \mu({\rm UC})$  & $\mu({\rm TC}) < \mu({\rm UC})$  & 		 &  0.02 & 0.01 &  0.00 &  0.00\\ 
	$\mu({\rm TC}) = \mu({\rm DVD})$  & $\mu({\rm TC}) \neq \mu({\rm DVD})$  &  & 0.18 &  0.00 &  0.00&  0.00 \\ 
		$\mu({\rm TC}) \geq \mu({\rm DVD})$  & $\mu({\rm TC}) < \mu({\rm DVD})$  & 	 &  0.09 &  0.00 &  0.00 &  0.00\\ 
				\hline
	\end{tabular}
		\end{table}
		
	The results in Table 4 show that TC is better than either UC or DVD with high significance when 
		$\hat\psi_{d_0}$, 
		$\hat\psi_{d_1}$, or  
		$\hat\psi_{d_2}$ is used, but not $\hat\psi_{p}$. Thus, applying SDR is beneficial in this example. 
		Also, all methods cannot detect any difference between DVD and UC. 
		
		The previous analysis shows some advantages of using the proposed 
		$\hat\psi_{d_1}$ and/or $\hat\psi_{d_2}$, but the accuracy of estimators of $E(Y|U=u_0,a)$ has not been 
		investigated. Different from the simulation study, 
	the true value of $E(Y|U=u_0, a)$ is unknown in the real data analysis. Thus, we apply the following cross-validation to assess the accuracy for different estimation methods. 
	
	The following discussion is for a fixed $a =$ DVD, TC, or UC.
	 We divide the dataset into 10 subsets with roughly the same sample size, say $n_c$. Let $S_c$ be one such subset, $c=1,...,10$. We use data not in $S_c$ but in other subsets to obtain the estimator $\hat\psi_d^{(c)}$, where supscript $(c)$ indicates using data not in $S_c$ and $d=p,d_0,d_1$, or $d_2$. 	Then, we assess the accuracy by $[ Y_i - \hat\psi_d^{(c)}(U_i)]^2$ for $i \in S_c$, noting that $(Y_i,U_i)$, $i \in S_c$, are not used in the construction of $\hat\psi_d^{(c)}$. After crossover all $S_c$, we 
	estimate $E[Y- \hat\psi_d(U)]^2$ by 
	\[ CV(10) = \frac{1}{10} \sum_{c=1}^{10} \frac{1}{n_c} \sum_{i \in S_c} \left[ Y_i - \hat\psi_d^{(c)} (U_i)\right]^2 .   \]
However,  $E[Y- \hat\psi_d(U)]^2$  is not the mean-squared error of $\hat\psi_d$. 
Note that 
$$
E[Y-\hat\psi_d (U) ]^2  =E[E(Y|U)-\hat\psi_d(U)  ]^2 
+ E[Y - E(Y|U) ]^2 
$$
because the cross product term 
\begin{align*}
E\left([E(Y|U)-\hat\psi_d(U) ] [Y - E(Y|U)] \right) & = 
E\left\{ E \left([E(Y|U)-\hat\psi_d(U) ] [Y - E(Y|U) ] \bigg| U\right) \right\} \\
& = E\left\{[E(Y|U)-\hat\psi_d(U)  ] E  \left( [Y - E(Y|U)] \bigg| U\right) \right\}\\
& = 0 .  
\end{align*} 
The term $E[E(Y|U)-\hat\psi_d(U)  ]^2$ is an average mean squared error (AMSE) of 
$\hat\psi_d(U)$ over all $U$ values. 
If we can estimate $\sigma^2 =  E[Y - E(Y|U) ]^2$ by $\hat\sigma^2$, then we can estimate AMSE of $\hat\psi_d$ by $CV(10)-\hat\sigma^2$. 

We utilize the difference-based variance estimators proposed in \citet{hall_kay_titterington_1990}, \citet{hall_kay_titterington_1991} and \citet{munk_bissantz_wagner_freitag_2005} to estimate $\sigma^2$, where $U$ in 
$\sigma^2 = E[Y-E(Y|U)]^2$ is 
replaced by    $\hat{B}^T U$.
Estimated values of $\sigma^2$, AMSE for
 $\hat\psi_{d_0}$,  $\hat\psi_{d_1}$, 
 and  $\hat\psi_{d_2}$ for three groups DVD, TC, and UC are  
 given in Table 5. 
 
 It can be seen from Table 5 that using $Z$ data helps in the estimation of $E(Y|U=u_0)$ on the average, especially when we use $\hat\psi_{d_2}$. 
 
\section*{Appendix}


\textbf{Proof of Lemma 1}
\begin{proof}
(i) Since $Y\independent (Z,U) \mid C_z^T Z + C_u^T U$, then $Y\independent (Z,U) \mid (C_z^T Z ,C_u^T U)$ and further $Y\independent U \mid (C_z^T Z ,C_u^T U)$. By the definition of partial central subspace in \citet{chiaromonte_cook_li_2002}, 
the partial central space $\mathcal{S}_{Y|U}^{C_z^T Z}\subseteq \mathcal{S}(C_u)$. By  Proposition 3.1 and equation (3.1) in \citet{hung_liu_lu_2015}, it is easy to get $\mathcal{S}(B)\subseteq \mathcal{S}_{Y|U}^{C_z^T Z} \oplus \mathcal{S}(C)\subseteq \mathcal{S}(C_u) \oplus \mathcal{S}(C)$. Furthermore, if $C_z^T Z \independent U \mid B^T U$, by the denfinition of central subspace in \citet{cook_1998}, $\mathcal{S}(C)\subseteq \mathcal{S}(B)$. \\
(ii) 
For the same reason as in (i), the partial central space $\mathcal{S}_{Y|B^TU}^{D_z^T Z}\subseteq \mathcal{S}(D_u)$, and then we can get $\mathcal{S}(I_{d})\subseteq \mathcal{S}_{Y|B^TU}^{D_z^T Z}\oplus \mathcal{S}(D)\subseteq \mathcal{S}(D_u)\oplus \mathcal{S}(D)$.  On the other hand, $\mathcal{S}(D_u)\oplus \mathcal{S}(D) \subseteq\mathcal{S}(I_d)$. As a result, $\mathcal{S}(D_u)\oplus \mathcal{S}(D) = \mathcal{S}(I_d)$, which means $\mathcal{S}(B) = \mathcal{S}(BD_u)\oplus \mathcal{S}(BD)$.
\end{proof}

The following Lemmas \ref{lemMaZhu} - \ref{lemlin} are all used for the proof of Theorem 1. 

\begin{lemma}\label{lemMaZhu}
If
Assumptions 
1,
3(i)(ii) and 
4 hold, then 
\begin{align*}
&\sup_{z, \Omega}\abs[\Bigg]{\dfrac{\hat{\gamma}_2(\hat{C}_z^T z + \hat{C}_u^T u_0)}{\hat{\gamma}_1(\hat{C}_z^T z + \hat{C}_u^T u_0)}-\dfrac{\hat{\gamma}_2(C_z^T z+C_u^T u_0)}{\hat{\gamma}_1(C_z^T z+C_u^T u_0)}-\dfrac{\gamma_2(\hat{C}_z^T z + \hat{C}_u^T u_0)}{\gamma_1(\hat{C}_z^T z + \hat{C}_u^T u_0 )}+\dfrac{\gamma_2(C_z^T z+C_u^T u_0)}{\gamma_1(C_z^T z+C_u^T u_0)}}\\
= & O_{p}(\hbar^{\mbar}n^{-1/2}+n^{-1}\hbar^{-(\dbar+1)}\log n),
\end{align*}
\begin{align*}
&\sup_{\Omega}\Bigg|\dfrac{1}{n}\sum_{i=1}^{n}\dfrac{\hat{\gamma}_2(\hat{C}_z^T Z_i+\hat{C}_u^T u_0)}{\hat{\gamma}_1(\hat{C}_z^T Z_i+\hat{C}_u^T u_0)}K_{\stwo{h}}(\hat{C}^T U_i-\hat{C}^T u_0) \\
& -\dfrac{1}{n}\sum_{i=1}^{n}\dfrac{\hat{\gamma}_2(\hat{C}_z^T Z_i +\hat{C}_u^T u_0)}{\hat{\gamma}_1(\hat{C}_z^T Z_i+\hat{C}_u^T u_0)}K_{\stwo{h}}(C^T U_i-C^T u_0)\\
&-E\left(\left.\dfrac{\hat{\gamma}_2(\hat{C}_z^T Z+\hat{C}_u^T u_0)}{\hat{\gamma}_1(\hat{C}_z^T Z+\hat{C}_u^T u_0)}\right\vert \hat{C}^T U = \hat{C}^T u_0\right) f_{\hat{C}^T U}(\hat{C}^T u_0)\\ 
& +E\left(\left.\dfrac{\hat{\gamma}_2(\hat{C}_z^T Z+\hat{C}_u^T u_0)}{\hat{\gamma}_1(\hat{C}_z^T Z+\hat{C}_u^T u_0)}\right\vert C^T U = C^T u_0\right)
f_{C^T U}(C^T u_0)\Bigg| \\
=& O_{p}(\stwo{h}^{\stwo{m}}n^{-1/2}+n^{-1}\stwo{h}^{-(\stwo{d}+1)})
\end{align*}
and
\begin{align*}
&\sup_{\Omega}\bigg|\dfrac{1}{n}\sum_{i=1}^{n} K_{\stwo{h}}(\hat{C}^T U_i-\hat{C}^T u_0)-\dfrac{1}{n}\sum_{i=1}^{n} K_{\stwo{h}}(C^T U_i-C^T u_0)-f_{\hat{C}^T U}(\hat{C}^T u_0)+ f_{C^T U}(C^T u_0)\bigg|\\
=& O_{p}(\stwo{h}^{\stwo{m}}n^{-1/2}+n^{-1}\stwo{h}^{-(\stwo{d}+1)})
\end{align*}
\end{lemma}
\begin{proof} 
The proof is analogous to that of Lemma 3 in the supplementary materials of \citet{ma_zhu_2012}. 
\end{proof}

Making use of the results in Lemma \ref{lemMaZhu}, we prove in the following Lemma \ref{lemSDRneg} that the estimation errors of SDR are asymptotically negligible. Note that we assume that SDR estimators of $C_{zu}$ and $C$ converge at the rate of $n^{-1/2}$. 
For the following lemmas and proofs, when necessary, we use $\hat{\psi}_{\stwo{d}}(u_0; \hat{\gamma}, \hat{C}_{zu}, \hat{C})$ to represent $\hat{\psi}_{\stwo{d}}(u_0)$ based on  $\hat{\gamma}$ and $(\hat{C}_{zu}, \hat{C})$, and use  $\hat{\psi}_{\stwo{d}}(u_0; \hat{\gamma}, C_{zu}, C)$  to represent $\hat{\psi}_{\stwo{d}}(u_0)$ based on  $\hat \gamma$ and $(C_{zu}, C)$.

\begin{lemma}\label{lemSDRneg}
If
Assumptions 
1, 
3(i)-(iii) and
4  hold, then 
\begin{align*}
 \sqrt{n\stwo{h}^{\stwo{d}}}\left[\hat{\psi}_{\stwo{d}}(u_0;\hat{\gamma}, \hat{C}_{zu},\hat{C})-\psi(u_0)\right]=\sqrt{n\stwo{h}^{\stwo{d}}}\left[\hat{\psi}_{\stwo{d}}(u_0;\hat{\gamma}, C_{zu}, C)-\psi(u_0)\right]+o_{p}(1)
 \end{align*}
\end{lemma}

\begin{proof} Write
\begin{equation*}\begin{split}
& \sqrt{n\stwo{h}^{\stwo{d}}}\left[\hat{\psi}_{\stwo{d}}(u_0;\hat{\gamma},\hat{C}_{zu},\hat{C})-\psi(u_0)\right]-\sqrt{n\stwo{h}^{\stwo{d}}}\left[\hat{\psi}_{\stwo{d}}(u_0;\hat{\gamma}, C_{zu},C)-\psi(u_0)\right]\\
=&\sqrt{n\stwo{h}^{\stwo{d}}}\left[\hat{\psi}_{\stwo{d}}(u_0;\hat{\gamma}, \hat{C}_{zu},\hat{C})-\hat{\psi}_{\stwo{d}}(u_0;\hat{\gamma}, \hat{C}_{zu},C)\right]+\sqrt{n\stwo{h}^{\stwo{d}}}\left[\hat{\psi}_{\stwo{d}}(u_0;\hat{\gamma}, \hat{C}_{zu}, C)-\hat{\psi}_{\stwo{d}}(u_0;\hat{\gamma}, C_{zu},C)\right]\\
=&R_{n1}+R_{n2}
\end{split}\end{equation*}
For $R_{n1}$, 
\begin{align*}
R_{n1}=&\sqrt{n\stwo{h}^{\stwo{d}}}\left[\dfrac{n^{-1}\sum_{i=1}^{n}\dfrac{\hat{\gamma}_2(\hat{C}_z^T Z_i + \hat{C}_u u_0)}{\hat{\gamma}_1(\hat{C}_z^T Z_i + \hat{C}_u u_0)}K_{\stwo{h}}(\hat{C}^T U_i - \hat{C}^T u_0)}{n^{-1}\sum_{i=1}^{n}K_{\stwo{h}}(\hat{C}^T U_i - \hat{C}^T u_0)}\right.\\
&\left.-\dfrac{n^{-1}\sum_{i=1}^{n}\dfrac{\hat{\gamma}_2(\hat{C}_z^T Z_i + \hat{C}_u u_0)}{\hat{\gamma}_1(\hat{C}_z^T Z_i + \hat{C}_u u_0)}K_{\stwo{h}}(C^T U_i - C^T u_0)}{n^{-1}\sum_{i=1}^{n}K_{\stwo{h}}\left(C^T U_i - C^T u_0\right)}\right]\\
= &\sqrt{n\stwo{h}^{\stwo{d}}}\left[n^{-1}\sum_{i=1}^{n}\dfrac{\hat{\gamma}_2(\hat{C}_z^T Z_i + \hat{C}_u u_0)}{\hat{\gamma}_1(\hat{C}_z^T Z_i + \hat{C}_u u_0)}K_{\stwo{h}}(\hat{C}^T U_i - \hat{C}^T u_0)\right.\\
&\left.-n^{-1}\sum_{i=1}^{n}\dfrac{\hat{\gamma}_2(\hat{C}_z^T Z_i + \hat{C}_u u_0)}{\hat{\gamma}_1(\hat{C}_z^T Z_i + \hat{C}_u u_0)}K_{\stwo{h}}(C^T U_i - C^T u_0)\right]\left[n^{-1}\sum_{i=1}^{n}K_{\stwo{h}}(\hat{C}^T U_i - \hat{C}^T u_0)\right]^{-1}\\
&-\sqrt{n\stwo{h}^{\stwo{d}}}n^{-1}\sum_{i=1}^{n}\dfrac{\hat{\gamma}_2(\hat{C}_z^T Z_i + \hat{C}_u u_0)}{\hat{\gamma}_1(\hat{C}_z^T Z_i + \hat{C}_u u_0)}K_{\stwo{h}}(C^T U_i - C^T u_0)\\
&\left[n^{-1}\sum_{i=1}^{n}K_{\stwo{h}}(\hat{C}^T U_i - \hat{C}^T u_0)-n^{-1}\sum_{i=1}^{n}K_{\stwo{h}}(C^T U_i - C^T u_0)\right]\\
&\left[n^{-1}\sum_{i=1}^{n}K_{\stwo{h}}(\hat{C}^T U_i - \hat{C}^T u_0)n^{-1}\sum_{i=1}^{n}K_{\stwo{h}}(C^T U_i - C^T u_0)\right]^{-1}\\
= &S_{n1}-S_{n2}
\end{align*}
$S_{n1}$ can be further split to 
\begin{align*}
S_{n1}=&\sqrt{n\stwo{h}^{\stwo{d}}}\left\{n^{-1}\sum_{i=1}^{n}\dfrac{\hat{\gamma}_2(\hat{C}_z^T Z_i + \hat{C}_u u_0)}{\hat{\gamma}_1(\hat{C}_z^T Z_i + \hat{C}_u u_0)}K_{\stwo{h}}(\hat{C}^T U_i - \hat{C}^T u_0)-n^{-1}\sum_{i=1}^{n}\dfrac{\hat{\gamma}_2(\hat{C}_z^T Z_i + \hat{C}_u u_0)}{\hat{\gamma}_1(\hat{C}_z^T Z_i + \hat{C}_u u_0)}\right.\\
&\left.K_{\stwo{h}}(C^T U_i - C^T u_0) -E\left[\left. \dfrac{\hat{\gamma}_2(\hat{C}_z^T Z + \hat{C}_u u_0)}{\hat{\gamma}_1(\hat{C}_z^T Z + \hat{C}_u u_0)}\right\vert\hat{C}^{T}U=\hat{C}^{T}u_0\right]f_{\hat{C}^T U}(\hat{C}^{T}u_0)\right.\\
&\left.+E\left[\left. \dfrac{\hat{\gamma}_2(\hat{C}_z^T Z + \hat{C}_u u_0)}{\hat{\gamma}_1(\hat{C}_z^T Z + \hat{C}_u u_0)}\right\vert C^{T}U=C^{T}u_0\right]f_{C^T U}(C^{T}u_0)\right\}\left[n^{-1}\sum_{i=1}^{n}K_{\stwo{h}}(\hat{C}^T U_i - \hat{C}^T u_0)\right]^{-1}\\
&+\sqrt{n\stwo{h}^{\stwo{d}}}\left\{E\left[\left. \dfrac{\hat{\gamma}_2(\hat{C}_z^T Z + \hat{C}_u u_0)}{\hat{\gamma}_1(\hat{C}_z^T Z + \hat{C}_u u_0)}\right\vert\hat{C}^{T}U=\hat{C}^{T}u_0\right]f_{\hat{C}^T U}(\hat{C}^{T}u_0)-E\Bigg[\dfrac{\hat{\gamma}_2(\hat{C}_z^T Z + \hat{C}_u u_0)}{\hat{\gamma}_1(\hat{C}_z^T Z + \hat{C}_u u_0)}\right.\\
&\left.\Bigg\vert C^{T}U=C^{T}u_0\Bigg]f_{C^T U}(C^{T}u_0)\right\} \left[n^{-1}\sum_{i=1}^{n}K_{\stwo{h}}(\hat{C}^T U_i - \hat{C}^T u_0)\right]^{-1}\\
=& S_{n11}+S_{n12}
\end{align*}
By 
Lemma \ref{lemMaZhu}, 
the numerator of $S_{n11}$ is bounded by
\begin{align*}
&\sqrt{n\stwo{h}^{\stwo{d}}}\sup_{\Omega}\left|n^{-1}\sum_{i=1}^{n}\dfrac{\hat{\gamma}_1(\hat{C}_z^T Z_i + \hat{C}_u u_0)}{\hat{\gamma}_1(\hat{C}_z^T Z_i + \hat{C}_u u_0)}K_{\stwo{h}}(\hat{C}^T U_i - \hat{C}^T u_0)-n^{-1}\sum_{i=1}^{n}\dfrac{\hat{\gamma}_1(\hat{C}_z^T Z_i + \hat{C}_u u_0)}{\hat{\gamma}_1(\hat{C}_z^T Z_i + \hat{C}_u u_0)}\right.\\
&\left. K_{\stwo{h}}(C^T U_i - C^T u_0) -E\left[\left. \dfrac{\hat{\gamma}_2(\hat{C}_z^T Z + \hat{C}_u u_0)}{\hat{\gamma}_1(\hat{C}_z^T Z + \hat{C}_u u_0)}\right\vert\hat{C}^{T}U=\hat{C}^{T}u_0\right]f_{\hat{C}^T U}(\hat{C}^{T}u_0)\right.\\
&\left. +E\left[\left. \dfrac{\hat{\gamma}_2(\hat{C}_z^T Z + \hat{C}_u u_0)}{\hat{\gamma}_1(\hat{C}_z^T Z + \hat{C}_u u_0)}\right\vert C^{T}U=C^{T}u_0\right]f_{C^T U}(C^{T}u_0)\right|\\
= &O_{p}\left(\sqrt{n\stwo{h}^{\stwo{d}}}(\stwo{h}^{\stwo{m}}n^{-1/2}+n^{-1}\stwo{h}^{-(\stwo{d}+1)})\right)\\
=& o_{p}(1)
\end{align*}
Also, the denominator of $S_{n11}$ converges to $f_{C^T U}(C^{T}u_0)$. Hence $S_{n11}=o_{p}(1)$. As to $S_{n12}$, by Lipschitz continuity in 
Assumption 
 3(iii), 
\begin{align*}
&\left|E\left[\left. \dfrac{\hat{\gamma}_2(\hat{C}_z^T Z + \hat{C}_u u_0)}{\hat{\gamma}_1(\hat{C}_z^T Z + \hat{C}_u u_0)}\right\vert\hat{C}^{T}U=\hat{C}^{T}u_0\right]f_{\hat{C}^T U}(\hat{C}^{T}u_0)\right.\\
&\left.-E\left[\left. \dfrac{\hat{\gamma}_2(\hat{C}_z^T Z + \hat{C}_u u_0)}{\hat{\gamma}_1(\hat{C}_z^T Z + \hat{C}_u u_0)}\right\vert C^{T}U=C^{T}u_0\right]f_{C^T U}(C^{T}u_0)\right|\leq \Lambda \norm{\hat{C}-C}
\end{align*} 
for some constant $\Lambda>0$, hence $\abs{S_{n12}}\leq O_{p}(\sqrt{n\stwo{h}^{\stwo{d}}}n^{-1/2})=o_{p}(1)$. Then $S_{n1}=o_{p}(1)$. The proof of $S_{n2}=o_{p}(1)$ is similar. As to $R_{n2}$, 
\begin{align*}
R_{n2}=&\sqrt{n\stwo{h}^{\stwo{d}}}\left[\dfrac{n^{-1}\sum_{i=1}^{n}\dfrac{\hat{\gamma}_2(\hat{C}_z^T Z_i + \hat{C}_u u_0)}{\hat{\gamma}_1(\hat{C}_z^T Z_i + \hat{C}_u u_0)}K_{\stwo{h}}(C^T U_i - C^T u_0)}{n^{-1}\sum_{i=1}^{n}K_{\stwo{h}}(C^T U_i - C^T u_0)}\right.\\
&\left.-\dfrac{n^{-1}\sum_{i=1}^{n}\dfrac{\hat{\gamma}_2(C_z^T Z_i + C_u^T u_0)}{\hat{\gamma}_1(C_z^T Z_i + C_u^T u_0)}K_{\stwo{h}}(C^T U_i - C^T u_0)}{n^{-1}\sum_{i=1}^{n}K_{\stwo{h}}\left(C^T U_i - C^T u_0\right)}\right]\\
=&\sqrt{n\stwo{h}^{\stwo{d}}}\Bigg\{n^{-1}\sum_{i=1}^{n}\left[\dfrac{\hat{\gamma}_2(\hat{C}_z^T Z_i + \hat{C}_u u_0)}{\hat{\gamma}_1(\hat{C}_z^T Z_i + \hat{C}_u u_0)}-\dfrac{\hat{\gamma}_2(C_z^T Z_i + C_u^T u_0)}{\hat{\gamma}_1(C_z^T Z_i + C_u^T u_0)}-\dfrac{\gamma_2(\hat{C}_z^T Z_i + \hat{C}_u u_0)}{\gamma_1(\hat{C}_z^T Z_i + \hat{C}_u u_0)}\right.\\
&\left.+\dfrac{\gamma_2(C_z^T Z_i + C_u^T u_0)}{\gamma_1(C_z^T Z_i + C_u^T u_0)}\right]K_{\stwo{h}}(C^T U_i - C^T u_0)\Bigg\}\left[n^{-1}\sum_{i=1}^{n}K_{\stwo{h}}(C^T U_i - C^T u_0)\right]^{-1}\\
&+\sqrt{n\stwo{h}^{\stwo{d}}}\left\{n^{-1}\sum_{i=1}^{n}\left[\dfrac{\gamma_2(\hat{C}_z^T Z_i + \hat{C}_u u_0)}{\gamma_1(\hat{C}_z^T Z_i + \hat{C}_u u_0)}-\dfrac{\gamma_2(C_z^T Z_i + C_u^T u_0)}{\gamma_1(C_z^T Z_i + C_u^T u_0)}\right]K_{\stwo{h}}(C^T U_i - C^T u_0)\right\}\\
&\left[n^{-1}\sum_{i=1}^{n}K_{\stwo{h}}(C^T U_i - C^T u_0)\right]^{-1}\\
=&S_{n3}+S_{n4}
\end{align*}
where $\abs{S_{n3}}\leq O_{p}(\sqrt{n\stwo{h}^{\stwo{d}}}(\hbar^{\mbar}n^{-1/2}+n^{-1}\hbar^{-(\dbar+1)}\log n))=o_{p}(1)$ by the uniform convergence result from 
Lemma \ref{lemMaZhu} and 
Assumption 
4. As to $S_{n4}$, 
\begin{align*}
&E\abs[\Bigg]{\left[\dfrac{\gamma_2(\hat{C}_z^T Z_i + \hat{C}_u u_0)}{\gamma_1(\hat{C}_z^T Z_i + \hat{C}_u u_0)}-\dfrac{\gamma_2(C_z^T Z_i + C_u^T u_0)}{\gamma_1(C_z^T Z_i + C_u^T u_0)}\right]K_{\stwo{h}}(C^T U_i - C^T u_0)}\\
\leq &\sup_{u,\norm{\hat{C}_{zu}-C_{zu}}\leq cn^{-1/2}}\Bigg\{E\left[\left. \norm[\bigg]{\dfrac{\partial}{\partial C_{zu}}E[Y|C_z^{T}Z+C_u^{T}u_0]}_{C_{zu}=C_{zu}+\delta (\hat{C}_{zu}-C_{zu})}\right\vert C^{T}U=C^{T}u\right]\\
&f_{C^T U}(C^{T}u)\Bigg\}\int \abs{K(\stwo{t})}\diff \stwo{t} \norm{\hat{C}_{zu}-C_{zu}}=O_{p}(\norm{\hat{C}_{zu}-C_{zu}})
\end{align*}
where $\delta\in (0,1)$. Thus $\abs{S_{n4}}\leq O_{p}(\sqrt{n\stwo{h}^{\stwo{d}}}n^{-1/2})=o_{p}(1)$.  Therefore, $R_{n1}$ and $R_{n2}$ are both $o_{p}(1)$. This completes the proof.
\end{proof}

By Lemma \ref{lemSDRneg},  the estimation errors of SDR estimators $(\hat{C}_{zu}, \hat{C})$ have no effect on the asymptotic distribution of 
 $\hat{\psi}_d(u_0;\hat\gamma, \hat C_{zu}, \hat C)$. Hence, in the following lemmas and proofs, we assume that  $(C_{zu}, C)$ are known.
The denominator of $\hat{\psi}_{\stwo{d}}(u_0; \hat{\gamma}, C_{zu}, C)$, 
 $n^{-1}\sum_{i=1}^n K_h(C^T U_i -C^T u_0)$, 
converges in probability to $f_{C^T U}(C^Tu_0)$. By Slutsky's theorem,  we only need to prove that the numerator of $\hat{\psi}_{\stwo{d}}(u_0; \hat{\gamma}, C_{zu}, C)$ is asymptotically normal. Write the numerator of $\hat{\psi}_{\stwo{d}}(u_0; \hat{\gamma}, C_{zu}, C)$ in the form of $n^{-1}\sum_{i=1}^n \psi_n(W_i, u_0; \hat \gamma)$ with $\psi_{n}(W_{i}, u_0 ;\hat{\gamma})=\hat{\varphi}_1( Z_i , u_0)K_{\stwo{h}}(C^TU_{i}-C^Tu_0)$ and $W=(Z,U)$. In the following proofs, we distinguish the kernels by denoting $\bar{K}$ and $K$ as the kernel functions  used in the first and second steps, respectively, and we use  $c$, $\tilde{c}$ and $M$ as generic constants.
The following Lemmas \ref{lemnonlin} - \ref{lemlin} together prove that 
$$\sqrt{n\stwo{h}^{\stwo{d}}}\dfrac{1}{n}\sum_{i=1}^{n}\Big[ \psi_{n}(W_{i}, u_0 ;\hat{\gamma})-\psi_{n}(W_{i}, u_0 ;\gamma)\Big]=o_{p}(1),$$
which means that the convergence rate of a two-step estimator is not directly affected by the kernel estimation of the inner layer, but by the kernel estimation of the outer layer.

\begin{lemma}\label{lemnonlin}
If
Assumptions 
1, 
2, 
3(ii) and 
4 hold, then 
\begin{align*}
T_{n1}=\sqrt{n\stwo{h}^{\stwo{d}}}\dfrac{1}{n}\sum_{i=1}^{n}\Big[ \psi_{n}(W_i, u_0 ;\hat{\gamma})-\psi_{n}(W_i, u_0 ;\gamma)-G_{n}(W_i, u_0 ;(\hat{\gamma}-\gamma))\Big]=o_{p}(1),
\end{align*}
where 
\begin{align*}
&G_{n}(W_{i}, u_0 ; \eta)=\dfrac{K_{\stwo{h}}(C^T U_i - C^T u_0)}{\gamma_1(C_z^T Z_i + C_u^T u_0 )}\left[\eta_2(C_z^T Z_i + C_u^T u_0)-\dfrac{\gamma_2(C_z^T Z_i + C_u^T u_0)}{\gamma_1(C_z^T Z_i + C_u^T u_0)}\eta_1 (C_z^T Z_i + C_u^T u_0)\right]
\end{align*}
with some functions $\eta_j(\cdot)$, $j=1,2$. 
\end{lemma}
\begin{proof}
From 
Assumption 
1, $\gamma_1$ is bounded away from zero. Since $\hat{\gamma}_1$ converges to $\gamma_1$ uniformly, when $n$ is large enough,  $\hat{\gamma}_1$ is also bounded away from zero, i.e. both $\inf \hat{\gamma}_1$ and $\inf \gamma_1\geq c$. Then by 
Assumptions 
2, 
3(ii) and 
4,
\begin{align*}
&\sqrt{n\stwo{h}^{\stwo{d}}}E\bigg[\abs[\Big]{\psi_{n}(W_{i}, u_0 ;\hat{\gamma})-\psi_{n}(W_{i}, u_0 ;\gamma)-G_{n}(W_{i}, u_0 ;(\hat{\gamma}-\gamma))}\bigg]\\
=&\sqrt{n\stwo{h}^{\stwo{d}}}E\left[\abs{K_{\stwo{h}}(C^T U_i - C^T u_0)}\dfrac{1}{\hat{\gamma}_{1}\gamma_1}\left(1+\abs[\bigg]{\dfrac{\gamma_2}{\gamma_1}}\right)\norm{\hat{\gamma}-\gamma}^2\right]\\
\leq &c^{-2}\sqrt{n\stwo{h}^{\stwo{d}}}E\left[\abs{K_{\stwo{h}}(C^T U_i - C^T u_0)}\left(1+\abs[\bigg]{\dfrac{\gamma_2}{\gamma_1}}\right)\norm{\hat{\gamma}-\gamma}^2 \right]\\
=&c^{-2} E\Big[\abs{K_{\stwo{h}}(C^T U_i - C^T u_0)}\left(1+\abs{E(Y|V=C_z^T Z_i + C_u^T u_0)}\right)\Big]\sqrt{n\stwo{h}^{\stwo{d}}}\norm{\hat{\gamma}-\gamma}_{\infty}^2\\
\leq & c^{-2}\sup_{u}\bigg\{ E\Big[ 1+\abs{E(Y|V=C_z^T Z_i + C_u^T u_0)} \mid C^T U_i=C^T u\Big]f_{C^T U}(C^T u)\bigg\}\int \abs{K(\stwo{t})}\diff \stwo{t}\\
& \sqrt{n\stwo{h}^{\stwo{d}}}\norm{\hat{\gamma}-\gamma}_{\infty}^2\\
=& o_{p}(1)
\end{align*}
By Chebyshev's inequality, $T_{n1}=o_{p}(1)$.
\end{proof}

\begin{lemma}\label{lemVStat}
Let $G_{n}(W_{i}, u_0 ;\eta)$ be as defined in 
Lemma \ref{lemnonlin}. If
Assumptions 
1, 
2, 
3(ii) (iv) and 
4 hold, then
\begin{align*}
T_{n2}=&\sqrt{n\stwo{h}^{\stwo{d}}}\left[\dfrac{1}{n}\sum_{i=1}^{n}G_{n}(W_i, u_0 ;(\hat{\gamma}-\gamma))-\int G_{n}(w, u_0 ;(\hat{\gamma}-\gamma))\diff F\right]=o_{p}(1)
\end{align*}
where $F$ is the cdf of $W$.
\end{lemma}
\begin{proof}
Let $\bar{\gamma}(v)=E[\hat{\gamma}(v)]$, then
\begin{align*}
T_{n2}=&\sqrt{n\stwo{h}^{\stwo{d}}}\left[\dfrac{1}{n}\sum_{i=1}^{n}G_{n}(W_i, u_0 ;(\hat{\gamma}-\gamma))-\int G_{n}(w, u_0 ;(\hat{\gamma}-\gamma))\diff F\right]\\
=&\sqrt{n\stwo{h}^{\stwo{d}}}\left[\dfrac{1}{n}\sum_{i=1}^{n}G_{n}(W_i, u_0 ;(\hat{\gamma}-\bar{\gamma}))-\int G_{n}(w, u_0 ;(\hat{\gamma}-\bar{\gamma}))\diff F\right]\\
&+\sqrt{n\stwo{h}^{\stwo{d}}}\left[\dfrac{1}{n}\sum_{i=1}^{n}G_{n}(W_i, u_0 ;(\bar{\gamma}-\gamma))-\int G_{n}(w, u_0 ;(\bar{\gamma}-\gamma))\diff F\right]\\
=&T_{n21}+T_{n22}
\end{align*}
We only need to prove that $$T_{n21}=\sqrt{n\stwo{h}^{\stwo{d}}}\left[\dfrac{1}{n}\sum_{i=1}^{n}G_{n}(W_i, u_0 ;(\hat{\gamma}-\bar{\gamma}))-\int G_{n}(w, u_0 ;(\hat{\gamma}-\bar{\gamma}))\diff F\right]$$ and $$T_{n22}=\sqrt{n\stwo{h}^{\stwo{d}}}\left[\dfrac{1}{n}\sum_{i=1}^{n}G_{n}(W_i, u_0 ;(\bar{\gamma}-\gamma))-\int G_{n}(w, u_0 ;(\bar{\gamma}-\gamma))\diff F\right]$$ are both $o_{p}(1)$.
As $T_{n21}$ can be written in such way that
$$T_{n21}=n^{-2}\sum_{i=1}^{n}\sum_{j=1}^{n}m_{n}({W_i,W_j})-n^{-1}\sum_{i=1}^{n}[m_{n1}(W_i)+m_{n2}(W_i)]+\mu$$
where 
\begin{align*}
&m_{n}(W_i,W_j)\\
&=\dfrac{K_{\stwo{h}}(C^T U_i -C^T u_0)}{\gamma_1(C_z^T Z_i + C_u^T u_0)}\left[Y_j-\dfrac{\gamma_2 (C_z^T Z_i + C_u^T u_0)}{\gamma_1(C_z^T Z_i + C_u^T u_0)}\right]\bar{K}_{\hbar}(C_z^T Z_i + C_u^T u_0-C_z^T Z_j - C_u^T U_j),
\end{align*}
$m_{n1}(W_i)=\int m_{n}(W_i,w)\diff F$, $m_{n2}(W_i)=\int m_{n}(w,W_i)\diff F$ and $\mu=E[m_{n}(W_1,W_2)]$.
Result from Lemma 8.4 in \citet{newey_mcfadden_1994} concerning V-statistics convergence is applied directly. Lemma 8.4 states that if $W_{1},W_{2}...,W_{n}$ are i.i.d then 
\begin{align*}
&n^{-2}\sum_{i=1}^{n}\sum_{j=1}^{n}m_{n}({W_i,W_j})-n^{-1}\sum_{i=1}^{n}[m_{n1}(W_i)+m_{n2}(W_i)]+\mu\\
&=O_{p}\left(n^{-1}E[\abs{m_{n}(W,W)}]+n^{-1}E^{1/2}[\abs{m_{n}(W_1,W_2)}^2]\right)
\end{align*}
In $T_{n21}$, 
\begin{align*}
&E[\abs{m_{n}(W,W)}]\\
\leq & E\left\{\abs[\bigg]{\dfrac{K_{\stwo{h}}(C^T U -C^T u_0)}{\gamma_1(C_z^T Z + C_u^T u_0)}\left[Y-\dfrac{\gamma_2(C_z^T Z + C_u^T u_0)}{\gamma_1(C_z^T Z + C_u^T u_0)}\right]\bar{K}_{\hbar}(0+(C_u^T U - C_u^T u_0))}\right\}\\
\leq & c^{-1} E\left\{\abs{K_{\stwo{h}}(C^T U -C^T u_0)\bar{K}_{\hbar}(C_u^T U -C_u^T u_0)}\left[\abs{Y}+\abs[\bigg]{\dfrac{\gamma_2(C_z^T Z + C_u^T u_0)}{\gamma_1(C_z^T Z + C_u^T u_0)}}\right]\right\}\\
= & c^{-1} E\bigg\{\abs{K_{\stwo{h}}(C^T U -C^T u_0)\bar{K}_{\hbar}(C_u^T U -C_u^T u_0)}\Big[E(\abs{Y}|C_u^T U,C^T U)\\
&+E\left(\abs{E(Y|V=C_z^T Z + C_u^T u_0)} \middle\vert C^T U\right)\Big]\bigg\}
\end{align*}
Note that $\mathcal{S}(C_u)$  and $\mathcal{S}(C)$ may overlap and it is necessary to find out the basis of $\mathcal{S}(C_{u})\oplus \mathcal{S}(C)$ before calculating the expectation.
Split the columns of $C_u$ and $C$ into two parts such that $C_u=(C_u^*,C_u^{**})$ and  $C=(C^*,C^{**})$, where the columns of $C_u^*$ and $C^*$ (with the smallest possible column dimensions) together form the basis of the space $\mathcal{S}(C_{u})\oplus \mathcal{S}(C)$. In this case, the columns of $C_u^{**}$(or $C^{**}$) can be written as linear combinations of the columns of $C_u^*$ and $C^*$.
Let $\dbar^{**}$ and $\stwo{d}^{**}$ be the column dimensions of $C_u^{**}$ and $C^{**}$ respectively. 
As the kernels are bounded from above, then there exists $M>0$ such that 
$\hbar^{-\dbar^{**}}\abs{\bar{K}\left((C_u^{**T} u-C_u^{**T} u_0)/\hbar\right)}
\leq \hbar^{-\dbar^{**}}M\leq \hbar^{-\dbar}M$ and 
$\stwo{h}^{-\stwo{d}^{**}}\abs{K\left((C^{**T} u-C^{**T} u_0)/\stwo{h}\right)}\leq \stwo{h}^{-\stwo{d}^{**}}M\leq \stwo{h}^{-\stwo{d}}M$. By 
Assumptions 
1 and 
3
(ii),
(iv),
\begin{align*}
 &E[\abs{m_{n}(W,W)}]\leq E\bigg\{\abs[\bigg]{\dfrac{K_{\stwo{h}}(C^T U - C^T u_0)}{\gamma_1(C_z^T Z+C_u^T u_0)}\left[Y-\dfrac{\gamma_2(C_z^T Z+C_u^T u_0)}{\gamma_1(C_z^T Z+C_u^T u_0)}\right]\bar{K}_{\hbar}(0+C_u^T U-C_u^T u_0)}\bigg\}\\
\leq  & c^{-1} \iint\abs{\bar{K}_{\hbar}(C_u^{*T} u-C_u^{*T} u_0)K_{\stwo{h}}(C^{*T} u-C^{*T} u_0)\bar{K}_{\hbar}(C_u^{**T} u-C_u^{**T} u_0)K_{\stwo{h}}(C^{**T} u-C^{**T} u_0)}\\
&\Big[E(\abs{Y}|C_u^T U=C_u^T u,C^T U=C^T u)+E\left(\abs{E(Y|V=C_z^T Z+C_u^T u_0)} \middle\vert C^T U=C^T u\right)\Big]\\
&f_{C_u^{*T}U, C^{*T}U}(C_u^{*T}u, C^{*T}u)\diff (C_u^{*T}u)\diff (C^{*T}u)\\
\leq & c^{-1}  M^2 \hbar^{-\dbar} h^{-\stwo{d}}\sup_{u}\bigg\{ \Big[E(\abs{Y}|C_u^{T} U=C_u^{T}u,C^{T} U=C^{T} u)+E\left(\abs{E(Y|V=C_z^T Z+C_u^T u_0)}\right. \\
&\left.\middle\vert C^{T} U=C^{T} u \right)\Big]f_{C_u^{*T} U,C^{*T} U}(C_u^{*T} u,C^{*T} u)\bigg\}\int \abs{\bar{K}(\bar{t}^*)}\diff \bar{t}^*\int\abs{K(\stwo{t}^*)}\diff \stwo{t}^*=O_p(\hbar^{-\dbar}h^{-\stwo{d}})
\end{align*}
where $\bar{t}^*\in\mathbb{R}^{\dbar-\dbar^{**}}$ and $\stwo{t}^*\in\mathbb{R}^{\stwo{d}-\stwo{d}^{**}}$.
%
As to $E[\abs{m_{n}(W_1,W_2)}^2]$,
\begin{align*}
& E\left[\abs{m_{n}(W_1,W_2)}^2\right]\\
\leq & E\Bigg\{\abs[\bigg]{\dfrac{K_{\stwo{h}}(C^T U_1-C^T u_0)}{\gamma_1(C_z^T Z_1+C_u^T u_0)}\left[Y_2-\dfrac{\gamma_2(C_z^T Z_1+C_u^T u_0)}{\gamma_1(C_z^T Z_1+C_u^T u_0)}\right]\bar{K}_{\hbar}(C_z^T Z_1+C_u^T u_0-V_2)}^2\Bigg\}\\
\leq & I_{n1}+I_{n2}
\end{align*}
where the first term
\begin{align*}
I_{n1}=&E\left\{\abs[\bigg]{\dfrac{K_{\stwo{h}}(C^T U_1-C^T u_0)}{\gamma_1(C_z^T Z_1+C_u^T u_0)}Y_2\bar{K}_{\hbar}(C_z^T Z_1+C_u^T u_0-V_2)}^2\right\}\\
\leq &c^{-2} E_{W_1}\bigg\{E_{W_2}\bigg[K^2_{\stwo{h}}(C^T U_1-C^T u_0)Y_2^2 \bar{K}_{\hbar}^2(C_z^T Z_1+C_u^T u_0-V_2)\bigg]\bigg\}\\
= &c^{-2}  E_{W_1}\bigg\{ K^2_{h}(C^T U_1-C^T u_0)E_{V_2}\bigg[E\left(Y_2^2|V_2\right)\bar{K}_{\hbar}^2(C_z^T Z_1+C_u^T u_0-V_2)\bigg]\bigg\}\\
  \leq &\tilde{c} \hbar^{-\dbar}h^{-\stwo{d}}  \sup_{v}\left\{E\left(Y^2|V=v\right)f_V(v)\right\}\sup_{u}\left\{f_{C^T U}(C^T u)\right\}\int \abs{\bar{K}(\bar{t})}\diff \bar{t} \int \abs{K(\stwo{t})} \diff \stwo{t}  \\
=&O_{p}(\hbar^{-\dbar}h^{-\stwo{d}})
\end{align*}
Similarly $I_{n2}\leq O_{p}(\hbar^{-\dbar}h^{-\stwo{d}})$. 
In all,  by 
Assumption 
4,
\begin{align*}
T_{n21}= &\sqrt{n\stwo{h}^{\stwo{d}}}O_{p}\left(n^{-1}E[\abs{m_{n}(W,W)}]+n^{-1}E[\abs{m_{n}(W_1,W_2)}^2]^{1/2}\right)\\
\leq & O_{p}\left(\sqrt{n\stwo{h}^{\stwo{d}}}\hbar^{-\dbar}h^{-\stwo{d}}n^{-1}\right)=o_{p}(1).
\end{align*}
As to $T_{n22}$, by Chebychev's Inequality,  
Assumptions 
1, 
3
(ii) and 
4, since $E(T_{n22})=0$,
\begin{align*}
&P\left(\abs[\Bigg]{\sqrt{n\stwo{h}^{\stwo{d}}}\left[\dfrac{1}{n}\sum_{i=1}^{n}G_{n}(W_{i},u_0; \bar{\gamma}-\gamma)-\int G_{n}(w,u_0;\bar{\gamma}-\gamma)\diff F\right]}>\epsilon\right)\\
=& P\left(\abs[\Bigg]{\sqrt{n\stwo{h}^{\stwo{d}}}\left[\dfrac{1}{n}\sum_{i=1}^{n}G_{n}(W_{i},u_0; \bar{\gamma}-\gamma)-\int G_{n}(w,u_0;\bar{\gamma}-\gamma)\diff F\right]}^2>\epsilon^2\right)\\
\leq & \dfrac{n\stwo{h}^{\stwo{d}}}{n^2\epsilon^2}\Bigg\{n(n-1)\bigg|E\bigg[G_{n}(W_{i},u_0; \bar{\gamma}-\gamma)-\int G_{n}(w,u_0;\bar{\gamma}-\gamma)\diff F\bigg]\bigg|^2\\
&+n E\left[\abs[\bigg]{G_{n}(W_i,u_0; \bar{\gamma}-\gamma)-\int G_{n}(w,u_0;\bar{\gamma}-\gamma)\diff F}^2\right]\Bigg\}\\
=&\stwo{h}^{\stwo{d}} \epsilon^{-2}E\left[\abs[\bigg]{G_{n}(W_{i},u_0; \bar{\gamma}-\gamma)-\int G_{n}(w,u_0;\bar{\gamma}-\gamma)\diff F}^2\right]\\
\leq & \stwo{h}^{\stwo{d}} \epsilon^{-2}E[\abs{G_{n}(W_{i},u_0; \bar{\gamma}-\gamma)}^2]\\
\leq & \stwo{h}^{\stwo{d}} \epsilon^{-2}E\left\{\dfrac{K^2_{h}(C^T U_{i}-C^T u_0)}{\gamma_1^2(C_z^T Z_{i}+C_u^T u_0)}\abs[\bigg]{\left[-\dfrac{\gamma_2(C_z^T Z_{i}+C_u^T u_0)}{\gamma_1(C_z^T Z_{i}+C_u^T u_0)},1\right] \left[\bar{\gamma}_1-\gamma_1,\bar{\gamma}_2-\gamma_2\right]}^2 \right\}\\
\leq & c^{-2} \epsilon^{-2} \stwo{h}^{\stwo{d}} E\left\{K^2_{h}(C^T U_{i}-C^T u_0)E\left[\left. \abs[\bigg]{\dfrac{\gamma_2(C_z^T Z_{i}+C_u^T u_0)}{\gamma_1(C_z^T Z_{i}+C_u^T u_0)}}^2+1\right\vert C^T U_{i}\right]\right\}\norm{\bar{\gamma}-\gamma_{0}}_{\infty}^2\\
 = & \tilde{c} \epsilon^{-2}\sup_{u}\bigg\{E\Big[\left. \abs{E[Y|V=C_z^T Z_{i}+C_u^T u_0]}^2+1\right\vert C^T U_{i}=C^T u\Big]f_{C^T U}(C^T u)\bigg\}\\
 &\int \abs{K(\stwo{t})}\diff \stwo{t} \norm{\bar{\gamma}-\gamma_{0}}_{\infty}^2\to_{p} 0,
\end{align*}
so $T_{n22}=o_{p}(1)$.
\end{proof}

For the following proof, we denote
$$s_{n}(V)=\int K_{\stwo{h}}(C^T u-C^T u_0)f_{C_z^T Z,C^T U}(V-C_u^T u_0, C^T u)\diff (C^T u) \enskip f_V^{-1}(V) [-E(Y|V),1]$$
and $S_{n}(Y,V)=s_{n}(V)(1,Y)^T$.

\begin{lemma}\label{lemlin}
If
Assumptions 
1, 
2, 
3
(ii) and 
4 hold, then
\begin{align*}
&T_{n3}=\sqrt{n\stwo{h}^{\stwo{d}}}\left[\int G_{n}(w, u_0; (\hat{\gamma}-\gamma))\diff F-\dfrac{1}{n}\sum_{j=1}^{n} S_{n}(Y_j, V_j)\right]=o_{p}(1)
\end{align*}
\end{lemma}
\begin{proof}


Since
\begin{align*}
&\sqrt{n\stwo{h}^{\stwo{d}}}\int G_{n}(w, u_0; \hat{\gamma}-\gamma)\diff F=\sqrt{n\stwo{h}^{\stwo{d}}}\int G_{n}(w, u_0; \hat{\gamma})\diff F\\
& = \sqrt{n\stwo{h}^{\stwo{d}}}\dfrac{1}{n}\sum_{j=1}^{n}\int S_{n}(Y_j,C_z^T z+C_u^T u_0) \bar{K}_{\hbar}(C_z^T z+C_u^T u_0-C_z^T Z_j+C_u^T U_j)\diff (C_z^T z),
\end{align*}
$T_{n3}$ can be written as 
\begin{align*}
T_{n3}=&\sqrt{n\stwo{h}^{\stwo{d}}}\dfrac{1}{n}\sum_{j=1}^{n}\int \left[S_{n}(Y_j, C_z^T z+C_u^T u_0)-S_{n}(Y_j, V_j)\right]\bar{K}_{\hbar}(C_z^T z+C_u^T u_0-V_j)\diff (C_z^T z)\\
=&\sqrt{n\stwo{h}^{\stwo{d}}}\dfrac{1}{n}\sum_{j=1}^{n}D_{n}(Y_j,V_j)
\end{align*}
by Chebyshev's inequality,
\begin{align*}
&P\left(\abs[\bigg]{\sqrt{n\stwo{h}^{\stwo{d}}}\dfrac{1}{n}\sum_{j=1}^{n}D_{n}(Y_j,V_j)}^2>\epsilon^2\right)\\
&\leq n \stwo{h}^{\stwo{d}}\bigg\{n(n-1)\abs{E[D_{n}(Y_j,V_j)]}^2+nE[D_{n}^2(Y_j,V_j)]\bigg\}\bigg/(n^2\epsilon^2),
\end{align*}
so we only need to prove
$\sqrt{n\stwo{h}^{\stwo{d}}}\abs{E[D_{n}(Y_j,V_j)]}\to 0$ and $\stwo{h}^{\stwo{d}}E[D_{n}^2(Y_j,V_j)]\to 0$.
By 
Assumptions 
1, 
2, 
3
(ii) and 
4, 
\begin{align*}
&\sqrt{n\stwo{h}^{\stwo{d}}}\abs[\bigg]{E[D_{n}(Y_j,V_j)]}\\
=&\sqrt{n\stwo{h}^{\stwo{d}}}\abs[\bigg]{E\left[\int S_{n}(Y_j,C_z^T z+C_u^T u_0)\bar{K}_{\hbar}(C_z^T z+C_u^T u_0-V_j)\diff (C_z^T z)\right]-E\left[S_{n}(Y_j,V_j)\right]}\\
=&\sqrt{n\stwo{h}^{\stwo{d}}}\left|\int s_{n}(C_z^T z+C_u^T u_0) E\bigg\{\bar{K}_{\hbar}(C_z^T z+C_u^T u_0-V_j)E\left[(1,Y_j)^{T}| V_j\right]\bigg\} \diff (C_z^T z)\right.\\
&\left.-E\bigg\{s_{n}(V_j)E[(1,Y_j)^{T}|V_j]\bigg\}\right|\\
=&\sqrt{n\stwo{h}^{\stwo{d}}}\left| \int s_{n}(v)\int \bar{K}(\bar{t}) \bigg\{E\left[(1,Y_j)^{T}| V_{j}=v+\hbar\bar{t}\right]f_{V}(v+\hbar\bar{t})\right.\\
&\left. - E[(1,Y_j)^{T}|V_{j}=v] f_{V}(v) \bigg\}\diff \bar{t}\diff v\right|\\
\leq &\sqrt{n\stwo{h}^{\stwo{d}}}\left| \int s_{n}(v)\sum_{k=1}^{\mbar-1}\hbar^{k}[k!]^{-1} \dfrac{\partial^{k} E[(1,Y_j)^{T}|V_{j}=v] f_{V}(v)}{\partial v^{k}} \int \bar{K}(\bar{t})[\otimes^{k} \bar{t}]\diff \bar{t}\diff v\right|\\
&+\sqrt{n\stwo{h}^{\stwo{d}}}\hbar^{\mbar}[\mbar!]^{-1}\int \norm{s_{n}(v)}\diff v \norm[\bigg]{\dfrac{\partial^{\mbar} E[(1,Y_j)^{T}|V_{j}=v] f_{V}(v)}{\partial v^{\mbar}}}_{\infty} \int\norm[\bigg]{\bar{K}(\bar{t})[\otimes^{\mbar} \bar{t}]}\diff \bar{t}\\
&=O_{p}(\sqrt{n\stwo{h}^{\stwo{d}}}\hbar^{\mbar})\to 0
\end{align*}
and by 
Assumptions 
1 and 
3
(ii),
\begin{align*}
&\stwo{h}^{\stwo{d}}E[D_{n}^2(Y_j,V_{j})]\\
&\leq \stwo{h}^{\stwo{d}}E\left[\abs[\bigg]{\int S_{n}(Y_j,C_z^T z+C_u^T u_0)\bar{K}_{\hbar}(C_z^T z+C_u^T u_0-V_{j})\diff (C_z^T z)}^2\right]+\stwo{h}^{\stwo{d}}E\left[S^2_{n}(Y_j,V_{j})\right]\\
& \leq \stwo{h}^{\stwo{d}}\sup_{\norm{\nu}<\epsilon} E[S_{n}^2(Y_j,V_{j}+\nu)]\left[\int \abs{\bar{K}(\bar{t})}\diff \bar{t}\right]^2+\stwo{h}^{\stwo{d}}E\left[S_{n}^2(Y_j,V_{j})\right]\\
&=O_{p}(\stwo{h}^{\stwo{d}})\to 0
\end{align*}
Thus, $T_{n3}=o_p(1)$ and the proof is completed.
\end{proof}
\subsection*{Proof of Theorem 1}
\begin{proof}
The proof follows the similar ideas as in \citet{newey_mcfadden_1994} and \citet{escanciano_et_al}.
From 
Lemma \ref{lemSDRneg}, 
\begin{align*}
 \sqrt{n\stwo{h}^{\stwo{d}}}\Big[\hat{\psi}_{\stwo{d}}(u_0;\hat{\gamma}, \hat{C}_{zu}, \hat{C})-\psi(u_0)\Big]=\sqrt{n\stwo{h}^{\stwo{d}}}\Big[\hat{\psi}_{\stwo{d}}(u_0;\hat{\gamma}, C_{zu}, C)-\psi(u_0)\Big]+o_{p}(1)
 \end{align*}
 Moreover, 
 \begin{align*}
 &\sqrt{n\stwo{h}^{\stwo{d}}}\left\{\hat{\psi}_{\stwo{d}}(u_0;\hat{\gamma},C_{zu}, C)-\psi(u_0)\right\}\\
 =&\sqrt{n\stwo{h}^{\stwo{d}}}\Bigg\{\dfrac{\frac{1}{n}\sum_{i=1}^{n}\dfrac{\hat{\gamma}_2(C_z^T Z_i+C_u^T u_0)}{\hat{\gamma}_1(C_z^T Z_i+C_u^T u_0)}K_{\stwo{h}}(C^T U_i-C^T u_0)}{\frac{1}{n}\sum_{i=1}^{n}K_{\stwo{h}}(C^T U_i-C^T u_0)}\\
 &-E[E(Y|V=C_z^T Z+C_u^T u_0) \mid C^T U=C^T u_0]\Bigg\}\\
 =&\sqrt{n\stwo{h}^{\stwo{d}}}\left\{\dfrac{1}{n}\sum_{i=1}^{n}\dfrac{\hat{\gamma}_2(C_z^T Z_i+C_u^T u_0)}{\hat{\gamma}_1(C_z^T Z_i+C_u^T u_0)}K_{\stwo{h}}(C^T U_i-C^T u_0)\right.\\
 &\left.-\dfrac{1}{n}\sum_{i=1}^{n}K_{\stwo{h}}(C^T U_i-C^T u_0)E[E(Y|V=C_z^T Z+C_u^T u_0)\right.\\
 &\mid C^T U=C^T u_0]\Bigg\}\left[\dfrac{1}{n}\sum_{i=1}^{n}K_{\stwo{h}}(C^T U_i-C^T u_0)\right]^{-1},
 \end{align*}
 where its numerator can be split as
\begin{align*}
&\sqrt{n\stwo{h}^{\stwo{d}}}\dfrac{1}{n}\sum_{i=1}^{n}K_{\stwo{h}}(C^T U_i-C^T u_0)\Bigg\{\dfrac{\hat{\gamma}_2(C_z^T Z_i+C_u^T u_0)}{\hat{\gamma}_1(\hat{C}_z^T Z_i+\hat{C}_u^T u_{0})}\\
&-E\left[E(Y|V=C_z^T Z+C_u^T u_0)\mid C^T U=C^T u_0\right]\Bigg\}\\
=&T_{n1}+T_{n2}+T_{n3}+\sqrt{n\stwo{h}^{\stwo{d}}}\dfrac{1}{n}\sum_{j=1}^{n} S_{n}(Y_j, V_j)\\
&+\sqrt{n\stwo{h}^{\stwo{d}}}\dfrac{1}{n}\sum_{i=1}^{n}K_{\stwo{h}}(C^T U_i-C^T u_0)\Bigg\{E(Y|V=C_z^T Z_i+C_u^T u_0)\\
&-E\left[E(Y|V=C_z^T Z+C_u^T u_0) \mid C^T U=C^T u_0\right]\Bigg\}
\end{align*}
By
Lemmas \ref{lemnonlin} - \ref{lemlin}, $T_{n1}$, $T_{n2}$ and $T_{n3}$ are all $o_{p}(1)$.  
It is also easy to prove 
$$\sqrt{n\stwo{h}^{\stwo{d}}}\dfrac{1}{n}\sum_{j=1}^{n} S_{n}(Y_j,V_j)=O_{p}(\sqrt{n\stwo{h}^{\stwo{d}}}n^{-1/2})=o_p(1)$$
 As a result, 
\begin{align*}
&\sqrt{n\stwo{h}^{\stwo{d}}}\dfrac{1}{n}\sum_{i=1}^{n}K_{\stwo{h}}(C^T U_i-C^T u_0)\bigg\{\dfrac{\hat{\gamma}_2(C_z^T Z_i+C_u^T u_0)}{\hat{\gamma}_1(C_z^T Z_i+C_u^T u_0)}\\
&-E\left[E(Y|V=C_z^T Z+C_u^T u_0)|C^T U=C^T u_0\right]\bigg\}\\
=&\sqrt{n\stwo{h}^{\stwo{d}}}\dfrac{1}{n}\sum_{i=1}^{n}K_{\stwo{h}}(C^T U_i-C^T u_0)\Big\{E[Y|V=C_z^T Z_i+C_u^T u_0]\\
&-E\left[E(Y|V=C_z^T Z+C_u^T u_0) \mid C^T U=C^T u_0\right]\Big\}+o_{p}(1)
\end{align*}
Follow the similar proof as in Theorem 2.2.2 of \citet{bierens}, when 
Assumptions 
3 
(v) and 
4 are satisfied,
\begin{equation*}\begin{split}
&n^{m/(2m+\stwo{d})}\bigg[\hat{\psi}_{\stwo{d}}(u_0;\hat{\gamma},\hat{C}_{zu}, \hat{C})-\psi(u_0)\bigg]\Longrightarrow N\left(\dfrac{\lambda b(C^T u_0)}{f_{C^T U}(C^T u_0)},\dfrac{g(C^T u_0)}{f_{C^T U}(C^T u_0)}\int K^2(\stwo{t})\diff \stwo{t} \right) 
\end{split}\end{equation*}
and its optimal convergence rate is $n^{-\stwo{m}/(2\stwo{m}+\stwo{d})}$. 

Specifically, when $\stwo{d}=d_1$, $\hat{\psi}_{\stwo{d}}(u_0)$ becomes $\hat{\psi}_{d_1}(u_0)$, 
(6) holds and its optimal convergence rate is $n^{-\stwo{m}/(2\stwo{m}+d_1)}$;  when $\stwo{d}=d_2$ and Assumptions 
1 - 
4 are satisfied with $C_z$, $C_u$, $C$ replaced by $D_z$, $BD_u$, $BD$ respectively, $\hat{\psi}_{\stwo{d}}(u_0)$ becomes $\hat{\psi}_{d_2}(u_0)$, 
(6) holds with $C_z$, $C_u$, $C$ replaced by $D_z$, $BD_u$, $BD$ respectively and $d=d_2$, and its optimal convergence rate is $n^{-\stwo{m}/(2\stwo{m}+d_2)}$.
\end{proof}

 \bibliographystyle{apalike}
 \bibliography{ref}

\begin{table}[h]
\centering
\caption{Values of $\hat\sigma^2$ and AMSE = $CV(10)- \hat\sigma^2$ 
for DVD, TC, and UC}\label{realdata_new_fixedB}
\begin{tabular}{cccccc}
  \hline
 & $\hat{\sigma}^2$ & AMSE($\hat{\psi}_p $)  & AMSE($\hat{\psi}_{d_0} $) & AMSE$(\hat{\psi}_{d_1}$) & AMSE($\hat{\psi}_{d_2}$)  \\ 
  \hline
DVD & 95.372 &  101.31 & 17.290 & 17.700 & 0.668 \\
TC & 86.803&   105.77& \ 4.394 & \ 4.944 & 3.799 \\  
UC & 101.35 &  157.94 & \ 3.358 & \ 7.453 & 2.789 \\
   \hline
\end{tabular}
\end{table}

\end{document}